%% file: Draft.tex
\documentclass[12pt,english]{article}
\usepackage[T1]{fontenc}
\usepackage[latin9]{inputenc}
\usepackage{geometry}
\usepackage{babel}
\usepackage{amsthm}
\usepackage{amsmath}
\usepackage{amssymb}
\usepackage{esint}
\usepackage[unicode=true,pdfusetitle,
 bookmarks=false,
 breaklinks=false,pdfborder={0 0 1},backref=false,colorlinks=false]
 {hyperref}
\hypersetup{
 pdfstartview={FitH}}

\makeatletter
\numberwithin{equation}{section}
\numberwithin{figure}{section}
 \theoremstyle{definition}
 \newtheorem*{defn*}{\protect\definitionname}
\theoremstyle{plain}
\newtheorem{thm}{\protect\theoremname}
  \theoremstyle{definition}
  \newtheorem{defn}[thm]{\protect\definitionname}
  \theoremstyle{plain}
  \newtheorem{prop}[thm]{\protect\propositionname}
  \theoremstyle{plain}
  \newtheorem{cor}[thm]{\protect\corollaryname}
  \theoremstyle{plain}
  \newtheorem{lem}[thm]{\protect\lemmaname}

\date{}

\makeatother

  \providecommand{\corollaryname}{Corollary}
  \providecommand{\definitionname}{Definition}
  \providecommand{\lemmaname}{Lemma}
  \providecommand{\propositionname}{Proposition}
\providecommand{\theoremname}{Theorem}

\begin{document}

\title{A modular spectral triple for $\kappa$-Minkowski space}

\author{Marco Matassa%
\thanks{SISSA, Via Bonomea 265, I-34136 Trieste, Italy\textit{. E-mail address}:
marco.matassa@sissa.it%
}}
\maketitle
\begin{abstract}
We present a spectral triple for $\kappa$-Minkowski space in two
dimensions. Starting from an algebra naturally associated to this
space, a Hilbert space is built using a weight which is invariant
under the $\kappa$-Poincaré algebra. The weight satisfies a KMS condition
and its associated modular operator plays an important role in the
construction. This forces us to introduce two ingredients which have
a modular flavor: the first is a twisted commutator, used to obtain
a boundedness condition for the Dirac operator, the second is a weight
replacing the usual operator trace, used to measure the growth of
the resolvent of the Dirac operator. We show that, under some assumptions
related to the symmetries and the classical limit, there is a unique
Dirac operator and automorphism such that the twisted commutator is
bounded. Then, using the weight mentioned above, we compute the spectral
dimension associated to the spectral triple and find that is equal
to the classical dimension. Finally we briefly discuss the introduction
of a real structure.

\newpage{}
\end{abstract}
\tableofcontents{}

\newpage{}

\section{Introduction}

In 1991 Lukierski, Ruegg, Nowicki \& Tolstoi \cite{luk1,luk2} introduced
a Hopf algebraic deformation of the Poincaré algebra, which is called
\textit{$\kappa$-Poincaré}, with a deformation parameter having units
of mass and denoted by $\kappa$. A few years later Majid and Ruegg
\cite{maj} clarified the bicrossproduct structure of $\kappa$-Poincaré:
it consists of a semidirect product of the classical Lorentz algebra,
which acts in a deformed way on the translation subalgebra, and a
backreaction of the momentum sector on the Lorentz transformations.
This allowed to introduce a homogeneous space for the $\kappa$-deformed
symmetry, as the quotient Hopf algebra of the $\kappa$-Poincaré group
with the Lorentz group. The result is a non-commutative Hopf algebra,
which can be interpreted as the algebra of functions over a non-commutative
spacetime, which is called \textit{$\kappa$-Minkowski}.

It is natural to ask whether $\kappa$-Minkowski is a non-commutative
geometry in the sense of Alain Connes \cite{connes}. Despite some
attempts \cite{dandrea-kappa,kappa1,kappa2}, it is fair to say that
a spectral triple that encodes the geometry of this space in a satisfactory
way has not yet been constructed. In particular, in none of the spectral
triples proposed so far the spectral dimension coincides with the
classical dimension, which is not very pleasant from the point of
view of the classical limit. In this paper we provide a new construction
of a spectral triple, which uses tools that have been developed to
study geometries with modular properties. They provide refinements
to the notion of spectral triple introduced by Connes, which is given
by the following.
\begin{defn*}
A compact spectral triple $(\mathcal{A},\mathcal{H},\mathcal{D})$
is the data of a unital $*$-algebra $\mathcal{A}$, a faithful $*$-representation
$\pi$ on a Hilbert space $\mathcal{H}$, and a self-adjoint operator
$\mathcal{D}$ such that\end{defn*}
\begin{itemize}
\item $[\mathcal{D},\pi(a)]$ extends to a bounded operator for all $a\in\mathcal{A}$,
\item $(\mathcal{D}-\mu)^{-1}$ is compact for all $\mu\notin\mbox{sp}\mathcal{D}$.
\end{itemize}
Here $\mbox{sp}\mathcal{D}$ is the spectrum of $\mathcal{D}$, and
we can rewrite the last condition in many equivalent forms, for example
requiring the compactness of the operator $(\mathcal{D}^{2}+1)^{-1/2}$.
One needs to modify this definition to treat non-compact cases, which
correspond to non-unital algebras. We refer to \cite{moyal} for a
discussion of the issues arising in this case and the necessary modifications.
Here we only ask for $\pi(a)(\mathcal{D}^{2}+1)^{-1/2}$ to be compact
for all $a\in\mathcal{A}$.

As a starting point for our construction we consider the $*$-algebra
introduced in \cite{DuSi11}, which we denote by $\mathcal{A}$, built
using the commutation relations associated to $\kappa$-Minkowski
space, and on which the $\kappa$-Poincaré algebra has a natural action.
It was also shown that the integral with respect to the Lebesgue measure
on $\mathbb{R}^{2}$, which we denote by $\omega$, is invariant under
the action of the $\kappa$-Poincaré algebra. It provides us with
the natural ingredient for the construction of a Hilbert space associated
to this geometry. We show that, using the GNS construction for $\omega$,
we obtain a Hilbert space which is unitarily equivalent to $L^{2}(\mathbb{R}^{2})$
and determine the corresponding unitary operator. We show that $\omega$
is a KMS weight for the algebra $\mathcal{A}$, and the corresponding
modular operator plays a major role in the rest of the construction.

The next step is the introduction of a self-adjoint operator $\mathcal{D}$,
which in the classical setting is given by the Dirac operator. In
the non-commutative setting this operator might be different from
the Dirac operator, that is the operator one would introduce to describe
the physical properties of fermions, but we still loosely refer to
$\mathcal{D}$ as such. We immediately face a difficulty in satisfying
the condition of boundedness of the commutator with $\mathcal{D}$,
which is related to the structure of the coproduct of the $\kappa$-Poincaré
algebra. One can relax this condition and consider the framework of
\textit{twisted spectral triples} \cite{type III}. This notion was
introduced to treat examples like perturbed spectral triples and the
transverse geometry of codimension one foliations. It was also remarked
that such a notion could be useful for spectral triples associated
to quantum groups. The novel ingredient of this framework is the twisted
commutator, defined by $[\mathcal{D},\pi(a)]_{\sigma}=\mathcal{D}\pi(a)-\pi(\sigma(a))\mathcal{D}$,
where $\sigma$ is an automorphism of the the algebra $\mathcal{A}$.
One asks that $[\mathcal{D},\pi(a)]_{\sigma}$ extends to a bounded
operator for all $a\in\mathcal{A}$.

We prove that, under some assumptions related to symmetry and to the
classical limit, there is a unique Dirac operator $\mathcal{D}$ and
a unique automorphism $\sigma$ such that the twisted commutator is
bounded. We also discuss the relations between $\mathcal{D}$, the
Casimir of the $\kappa$-Poincaré algebra and the equivariant Dirac
operator which has been considered in the literature.

At this point we arrive to one of the main results of the paper, the
computation of the spectral dimension. If this quantity exists we
say that the spectral triple is finitely summable, and this implies
the compactness requirement in the definition of a spectral triple.
However it turns out that, for our construction, the spectral dimension
does not exists in the usual sense of spectral triples. We are going
to argue that this problem arises from a mismatch of the modular properties
of the weight $\omega$, which we defined on our algebra $\mathcal{A}$,
and the non-commutative integral defined by the trace on the Hilbert
space.

A natural question at this point is whether we took a wrong turn in
our construction or if there is some way to make sense of this problem.
This is actually a quite general problem for geometries involving
KMS states (or weights, as in our case), and to deal with it the framework
of \textit{modular spectral triples} was introduced, see \cite{modular1,modular2,modular3,modular4}.
Unfortunately our model does not fit in this framework, which has
been formalized on the basis of examples where the modular group associated
to the KMS state is periodic, see for example the definition given
in \cite{modular4}. In this situation, it makes sense to restrict
some conditions to the fixed point algebra under this action. However
this possibility is not available in our case which, as we are going
to see, has an action of $\mathbb{R}$ by translation: indeed in this
case the only fixed point of $\mathcal{A}$ is given by the zero function.
It is possible that the framework might be extended to deal with this
class of examples, but we do not dwell on this point here.

Nevertheless, we show that we can borrow fruitfully some ingredients
from the modular spectral triples framework, and that they give interesting
results when applied to our case. For this reason we are going to
refer to this construction, loosely speaking, as a modular spectral
triple. In particular we use the notion of spectral dimension: to
this end we consider a weight $\Phi$, whose modular group is related
to that of $\omega$, and use it to compute the spectral dimension.
This weight has the role of fixing the mismatch in the modular properties
which we mentioned above. We show that, using this definition, our
spectral triple is finitely summable and its spectral dimension coincides
with the classical dimension. Moreover, if we take the residue at
$s=2$, that is the spectral dimension, of the {}``zeta function''
given by $\Phi(\pi(f)(\mathcal{D}^{2}+\mu^{2})^{-s/2})$, with $f\in\mathcal{A}$
and $\mu\neq0$, we recover the weight $\omega$ up to a constant.

This is the result one would expect, namely to recover the natural
notion of integration on the algebra, given by the weight $\omega$,
via the non-commutative integral, defined in terms of $\Phi$. Therefore
we see that, using ingredients from the modular spectral triples framework,
we can make sense of some of the constructions in the non-commutative
geometry approach. These results provide some preliminary evidences
that these are the right tools to describe the geometry of $\kappa$-Minkowski.
Finally we discuss the introduction of a real structure by defining
an antilinear isometry $\mathcal{J}$ on the Hilbert space. We check
the conditions which define a real structure in the classical case,
and show that they are modified. The most interesting modification
is given by the commutation relation between $\mathcal{D}$ and $\mathcal{J}$,
which involves the antipode structure of the $\kappa$-Poincaré algebra.

\section{The $*$-algebra}

\input{Algebra.tex}

\section{The Hilbert space}

\input{Hilbert.tex}

\section{The Dirac operator}

\input{Dirac.tex}

\section{The spectral dimension}

\input{Trace.tex}

\section{The real structure}

\input{Real.tex}

\section*{Acknowledgements}

I am indebted to Gherardo Piacitelli and Ludwik D\k{a}browski for
their assistance in the preparation of this paper. I also want to
thank Francesca Arici, Fabian Belmonte, and Domenico Monaco for helpful
discussions and comments.

\end{document}

%% file: Algebra.tex
The aim of this section is to introduce a $*$-algebra, which provides
the first ingredient for a spectral triple describing the geometry
of $\kappa$-Minkowski space. We start by recalling some basic facts
about the $\kappa$-Poincaré and $\kappa$-Minkowski Hopf algebras,
and some notions related to the implementation of Hopf algebra symmetries
on a Hilbert space. After this short review we describe the $*$-algebra
$\mathcal{A}$, which was introduced in \cite{DuSi11}, and recall
some of its properties which are relevant for the construction of
a spectral triple.

\subsection{The $\kappa$-Poincaré and $\kappa$-Minkowski algebras}

In this subsection we summarize the algebraic properties of the $\kappa$-Poincaré
algebra $\mathcal{P}_{\kappa}$ in two dimensions. First we give the
usual presentation that appears in the literature, that is as the
Hopf algebra generated by the elements $P_{0}$, $P_{1}$, $N$ satisfying
\[
\begin{split} & [P_{0},P_{1}]=0\ ,\quad[N,P_{0}]=P_{1}\ ,\\
 & [N,P_{1}]=\frac{\kappa}{2}(1-e^{-2P_{0}/\kappa})-\frac{1}{2\kappa}P_{1}^{2}\ .
\end{split}
\]
The coproduct $\Delta:\mathcal{P}_{\kappa}\to\mathcal{P}_{\kappa}\otimes\mathcal{P}_{\kappa}$
is defined by the relations
\[
\begin{split} & \Delta(P_{0})=P_{0}\otimes1+1\otimes P_{0}\ ,\qquad\Delta(P_{1})=P_{1}\otimes1+e^{-P_{0}/\kappa}\otimes P_{1}\ ,\\
 & \Delta(N)=N\otimes1+e^{-P_{0}/\kappa}\otimes N\ .
\end{split}
\]
The counit $\varepsilon:\mathcal{P}_{\kappa}\to\mathbb{C}$ and antipode
$S:\mathcal{P}_{\kappa}\to\mathcal{P}_{\kappa}$ are defined by
\[
\begin{split} & \varepsilon(P_{0})=\varepsilon(P_{1})=0\ ,\quad\varepsilon(N)=0\ ,\\
 & S(P_{0})=-P_{0}\ ,\quad S(P_{1})=-e^{P_{0}/\kappa}P_{1}\ ,\quad S(N)=-e^{P_{0}/\kappa}N\ .
\end{split}
\]
An important role is played by the Hopf subalgebra generated by $P_{0}$
and $P_{1}$, which we denote by $\mathcal{T}_{\kappa}$, that is
the generators of the translations. Indeed the $\kappa$-Minkowski
space is defined as the dual Hopf algebra to this subalgebra \cite{maj},
which we denote by $\mathcal{M}_{\kappa}$. If we denote the pairing
by $\langle\cdot,\cdot\rangle:\mathcal{T}_{\kappa}\times\mathcal{M}_{\kappa}\to\mathbb{C}$,
then the structure of $\mathcal{M}_{\kappa}$ is determined by the
duality relations 
\[
\begin{split}\langle t,xy\rangle & =\langle t^{(1)},x\rangle\langle t^{(2)},y\rangle\ ,\\
\langle ts,x\rangle & =\langle t,x^{(1)}\rangle\langle s,x^{(2)}\rangle\ .
\end{split}
\]
Here we have $t,s\in\mathcal{T}_{\kappa}$, $x,y\in\mathcal{M}_{\kappa}$
and we use the Sweedler notation for the coproduct
\[
\Delta x=\sum_{i}x_{(i)}^{(1)}\otimes x_{(i)}^{(2)}=x^{(1)}\otimes x^{(2)}\ .
\]
From the pairing we deduce that $\mathcal{M}_{\kappa}$ is non-commutative,
since $\mathcal{T}_{\kappa}$ is not cocommutative, that is the coproduct
in $\mathcal{T}_{\kappa}$ is not trivial. On the other hand, since
$\mathcal{T}_{\kappa}$ is commutative we have that $\mathcal{M}_{\kappa}$
is cocommutative. The algebraic relations for the $\kappa$-Minkowski
Hopf algebra $\mathcal{M}_{\kappa}$ are

\[
[X_{0},X_{1}]=-\kappa^{-1}X_{1}\ ,\qquad\Delta X_{\mu}=X_{\mu}\otimes1+1\otimes X_{\mu}\ .
\]

This concludes the usual presentation of the $\kappa$-Poincaré and
$\kappa$-Minkowski Hopf algebras. To be more precise such a presentation
is that of a $h$-adic Hopf algebra, where $h$ is the formal deformation
parameter, since we need to make sense of the power series of elements
like $e^{-P_{0}/\kappa}$. In this paper we prefer to use a slighly
different presentation, as done in \cite{DuSi11}, which gives a genuine
Hopf algebra. Instead of considering the exponential $e^{-P_{0}/\kappa}$
as a power series in $P_{0}$, we consider it as an invertible element
$\mathcal{E}$ and rewrite the defining relations as
\[
\begin{split} & [P_{0},P_{1}]=0\ ,\quad[P_{0},\mathcal{E}]=[P_{1},\mathcal{E}]=0\ ,\\
 & \Delta(P_{0})=P_{0}\otimes1+1\otimes P_{0}\ ,\quad\Delta(P_{1})=P_{1}\otimes1+\mathcal{E}\otimes P_{1}\ ,\quad\Delta(\mathcal{E})=\mathcal{E}\otimes\mathcal{E}\ ,\\
 & \varepsilon(P_{0})=\varepsilon(P_{1})=0\ ,\quad\varepsilon(\mathcal{E})=1\ ,\\
 & S(P_{0})=-P_{0}\ ,\quad S(P_{1})=-\mathcal{E}^{-1}P_{1}\ ,\quad S(\mathcal{E})=\mathcal{E}^{-1}\ .
\end{split}
\]
In this form the role of the subalgebra $\mathcal{T}_{\kappa}$ is
played by the one generated by $P_{\mu}$ and $\mathcal{E}$, which
we call the \textit{extended momentum algebra} and denote again by
$\mathcal{T}_{\kappa}$. An appropriate pairing defining $\kappa$-Minkowski
space can be easily written in terms of these generators. More importantly
it can be made into a Hopf $*$-algebra by defining the involution
as $P_{\mu}^{*}=P_{\mu}$ and $\mathcal{E}^{*}=\mathcal{E}$.

In the following we are going to consider the case of Euclidean signature
and so, strictly speaking, we should refer to the Euclidean counterpart
of the $\kappa$-Poincaré algebra, which is known as the \textit{quantum
Euclidean group}. However the boost generator $N$ is not going to
play a central role in our discussion, which is going to be based
on the extended momentum algebra, and therefore most of our relations
do not depend on the signature. Henceforth we only make reference
to the $\kappa$-Poincaré algebra and make some remarks when needed.

One more remark on the notation: we are going to write all formulae
in terms of the parameter $\lambda:=\kappa^{-1}$, instead of $\kappa$.
The motivation comes from the fact that the Poincaré algebra is obtained
in the {}``classical limit'' $\lambda\to0$, in a similar fashion
to the classical limit $\hbar\to0$ of quantum mechanics. This makes
more transparent checking that some formulae reduce, in this limit,
to their respective undeformed counterparts.

\subsection{Implementation of Hopf symmetries}

Here we recall some notions related to the implementation of Hopf
algebra symmetries on a Hilbert space \cite{equivariant}. In the
following $H$ denotes a Hopf algebra.
\begin{defn}
An \textit{$H$-module} is a pair $(V,\varphi)$, where $V$ is a
linear space and $\varphi$ is a complex linear representation of
$H$ on $V$. The representation will be also denoted by $h\triangleright v$,
where $h\in H$ and $v\in V$, instead of $\varphi(h)v$.
\end{defn}
In particular we are interested in the case in which $H$ acts on
an algebra $A$. The following definitions are compatibility conditions
between the two structures.
\begin{defn}
An algebra $A$ is a \textit{left $H$-module algebra} if $A$ is
a left $H$-module and the representation respects the algebra structure
of $A$, that is $h\triangleright(ab)=(h_{(1)}\triangleright a)(h_{(2)}\triangleright b)$
for all $h\in H$ and $a,b\in A$.
\end{defn}

\begin{defn}
A $*$-algebra $A$ is \textit{left $H$-module $*$-algebra} if $A$
is a left $H$-module algebra and moreover the action is compatible
with the star structure, that is $(h\triangleright a)^{*}=S(h)^{*}\triangleright a^{*}$
for all $h\in H$ and $a\in A$.
\end{defn}
We are interested in a representation of $H$ on a Hilbert space $\mathcal{H}$.
This representation will be unbounded in general, so we will have
to specify the appropriate domains.
\begin{defn}
Let $\mathcal{H}_{0}$ be a dense linear subspace of a Hilbert space
$\mathcal{H}$ with inner product $(\cdot,\cdot)$. An \textit{unbounded
$*$-representation} of $H$ on $\mathcal{H}_{0}$ is a homomorphism
$\rho:H\to\mbox{End}(\mathcal{H}_{0})$ such that $(\rho(h)\psi,\phi)=(\psi,\rho(h^{*})\phi)$
for all $\psi,\phi\in\mathcal{H}_{0}$ and $h\in H$.
\end{defn}
Now suppose that we have a representation $\pi$ of a $*$-algebra
$A$ on $\mathcal{H}$. Moreover suppose that $A$ is a left $H$-module
$*$-algebra, which is the case of interest for a spectral triple
with symmetries. We want the representation $\pi$ to be compatible
with the structure of the Hopf algebra $H$. This leads to the notion
of equivariance, given in the next definition.
\begin{defn}
Suppose $A$ is a left $H$-module $*$-algebra. A $*$-representation
$\pi$ of $A$ on $\mathcal{H}_{0}$ is called \textit{$H$-equivariant}
(or also covariant) if there exists an unbounded $*$-representation
$\rho$ of $H$ on $\mathcal{H}_{0}$ such that
\begin{equation}
\rho(h)\pi(a)\psi=\pi(h_{(1)}\triangleright a)\rho(h_{(2)})\psi\ ,
\end{equation}
for all $h\in H$, $a\in A$ and $\psi\in\mathcal{H}_{0}$.
\end{defn}
Finally we give a definition of equivariance for operators in $\mathcal{H}$.
\begin{defn}
A linear operator $T$ defined on $\mathcal{H}_{0}$ is said to be
\textit{equivariant} if it commutes with $\rho(h)$, that is $T\rho(h)\psi=\rho(h)T\psi$
for all $h\in H$ and $\psi\in\mathcal{H}_{0}$.
\end{defn}

\subsection{Definition of the $*$-algebra $\mathcal{A}$}

Now we can describe the $*$-algebra introduced in \cite{DuSi11},
see also \cite{dabro-piac}. We summarize some of the main results
and fix the notation for the other sections. In two dimensions the
underlying algebra of $\kappa$-Minkowski is the enveloping algebra
of the Lie algebra with generators $iX_{0}$ and $iX_{1}$ fullfilling
$[X_{0},X_{1}]=i\lambda X_{1}$. It has a faithful representation
$\varphi$ given by
\[
\varphi(iX_{0})=\left(\begin{array}{cc}
-\lambda & 0\\
0 & 0
\end{array}\right)\ ,\qquad\varphi(iX_{1})=\left(\begin{array}{cc}
0 & 1\\
0 & 0
\end{array}\right)\ .
\]
The corresponding simply connected Lie group $G$ consists of $2\times2$
matrices of the form
\begin{equation}
S(a)=S(a_{0},a_{1})=\left(\begin{array}{cc}
e^{-\lambda a_{0}} & a_{1}\\
0 & 1
\end{array}\right)\ ,\label{eq:param}
\end{equation}
which are obtained by exponentiating $\varphi$ as follows
\[
e^{i\varphi(a_{0}X_{0}+a_{1}^{\prime}X_{1})}=\left(\begin{array}{cc}
e^{-\lambda a_{0}} & \frac{1-e^{-\lambda a_{0}}}{\lambda a_{0}}a_{1}^{\prime}\\
0 & 1
\end{array}\right)\ .
\]
The group operations written in the $(a_{0},a_{1})$ coordinates are
given by
\[
S(a_{0},a_{1})S(a_{0}^{\prime},a_{1}^{\prime})=S(a_{0}+a_{0}^{\prime},a_{1}+e^{-\lambda a_{0}}a_{1}^{\prime})\ ,\qquad S(a_{0},a_{1})^{-1}=S(-a_{0},-e^{\lambda a_{0}}a_{1})\ .
\]

We have that the Lebesgue measure $d^{2}a$ is right invariant whereas
the measure $e^{\lambda a_{0}}d^{2}a$ is left invariant on $G$,
so in particular $G$ is not unimodular. We denote by $L^{1}(G)$
the convolution algebra of $G$ with respect to the right invariant
measure. We identify functions on $G$ with functions on $\mathbb{R}^{2}$
by the parametrization (\ref{eq:param}). Then $L^{1}(G)$ is an involutive
Banach algebra consisting of integrable functions on $\mathbb{R}^{2}$
with product $\hat{\star}$ and involution $\hat{*}$ given by
\[
\begin{split}(f\hat{\star}g)(a) & =\int f(a_{0}-a_{0}^{\prime},a_{1}-e^{-\lambda(a_{0}-a_{0}^{\prime})}a_{1}^{\prime})g(a_{0}^{\prime},a_{1}^{\prime})d^{2}a^{\prime}\ ,\\
f^{\hat{*}}(a) & =e^{\lambda a_{0}}\overline{f}(-a_{0},-e^{\lambda a_{0}}a_{1})\ .
\end{split}
\]
Any unitary representation $\tilde{\pi}$ of $G$ (assumed to be strongly
continuous) gives rise to a representation of $L^{1}(G)$, denoted
with the same symbol, obtained by setting
\[
\tilde{\pi}(f)=\int f(a)\tilde{\pi}(S(a))d^{2}a\ .
\]
It is indeed a $*$-representation, since it obeys the relations
\[
\tilde{\pi}(f\hat{\star}g)=\tilde{\pi}(f)\tilde{\pi}(g)\ ,\qquad\tilde{\pi}(f^{\hat{*}})=\tilde{\pi}(f)^{\dagger}\ ,
\]
for any $f,g\in L^{1}(G)$, where here $\dagger$ denotes the adjoint.
Now, following the same procedure as in the Weyl quantization, we
can define the map $W_{\tilde{\pi}}(f):=\tilde{\pi}(\mathcal{F}f)$,
where $f\in L^{1}(\mathbb{R}^{2})\cap\mathcal{F}^{-1}(L^{1}(\mathbb{R}^{2}))$
and $\mathcal{F}$ denotes the Fourier transform on $\mathbb{R}^{2}$.
It then follows that
\[
W_{\tilde{\pi}}(f\star g)=W_{\tilde{\pi}}(f)W_{\tilde{\pi}}(g)\ ,\qquad W_{\tilde{\pi}}(f^{*})=W_{\tilde{\pi}}(f)^{\dagger}\ ,
\]
where the product $\star$ and the involution $*$ are defined by
\begin{equation}
f\star g=\mathcal{F}^{-1}(\mathcal{F}f\hat{\star}\mathcal{F}g)\ ,\qquad f^{*}=\mathcal{F}^{-1}(\mathcal{F}f)^{\hat{*}}\ .\label{eq:alg-ops}
\end{equation}
Here the formulae are given using the unitary convention for the Fourier
transform.

It is important to remark that, although $W_{\tilde{\pi}}$ depends
on the choice of $\tilde{\pi}$ (and therefore on the choice of the
unitary representation of $G$), the product $\star$ and the involution
$*$ do not depend on such a choice. This is the same strategy which
is employed to define the deformed product in quantum mechanics. Indeed,
as in that case, one needs to exercise care about the domain of definition
of the operations $\star$ and $*$, so in the following we restrict
ourselves to a subset of Schwartz functions which have nice compatibility
properties with those operations.
\begin{defn}
Denote by $\mathcal{S}_{c}$ the space of Schwartz functions on $\mathbb{R}^{2}$
with compact support in the first variable, that is for $f\in\mathcal{S}_{c}$
we have $\mbox{supp}(f)\subseteq K\times\mathbb{R}$ for some compact
$K\subset\mathbb{R}$. We define $\mathcal{A}=\mathcal{F}(\mathcal{S}_{c})$,
where $\mathcal{F}$ is the Fourier transform on $\mathbb{R}^{2}$.\end{defn}
\begin{prop}
For $f,g\in\mathcal{A}$ we can write (\ref{eq:alg-ops}) in the following
form
\begin{equation}
\begin{split}(f\star g)(x) & =\int e^{ip_{0}x_{0}}(\mathcal{F}_{0}f)(p_{0},x_{1})g(x_{0},e^{-\lambda p_{0}}x_{1})\frac{dp_{0}}{2\pi}\ ,\\
f^{*}(x) & =\int e^{ip_{0}x_{0}}(\mathcal{F}_{0}\overline{f})(p_{0},e^{-\lambda p_{0}}x_{1})\frac{dp_{0}}{2\pi}\ .
\end{split}
\label{eq:alg-formulae}
\end{equation}
We have that $f\star g\in\mathcal{A}$ and $f^{*}\in\mathcal{A}$,
so that $\mathcal{A}$ is a $*$-algebra.
\end{prop}
Here $\mathcal{F}_{0}$ denotes the Fourier transform of $f$ in the
first variable, defined as 
\[
(\mathcal{F}_{0}f)(p_{0},x_{1})=\int e^{-ip_{0}y_{0}}f(y_{0},x_{1})dy_{0}\ .
\]
Notice that for $\lambda=0$ we recover the pointwise product $(f\star g)(x)=f(x)g(x)$
and the complex conjugation $f^{*}(x)=\overline{f}(x)$, giving the
correct classical limit for the algebra. The definition of $\mathcal{A}$
implies that any $f\in\mathcal{A}$ is a Schwartz function, but does
not have compact support in the first variable. However, by the Paley-Wiener
theorem, the compact support of $\mathcal{F}f$ in the first variable
implies analiticity of $f$ in the first variable.

The algebra $\mathcal{A}$ and the extended momentum algebra $\mathcal{T}_{\kappa}$
are naturally related. The main result that we need for our construction
is the following \cite{DuSi11}.
\begin{thm}
\label{prop:tk-module}The algebra $\mathcal{A}$ is a left $\mathcal{T}_{\kappa}$-module
$*$-algebra with respect to the following representation of the extended
momentum algebra
\[
(P_{\mu}\triangleright f)(x)=-i(\partial_{\mu}f)(x)\ ,\qquad(\mathcal{E}\triangleright f)(x)=f(x_{0}+i\lambda,x_{1})\ .
\]

\end{thm}
We recall that this means that, for any $f,g\in\mathcal{A}$ and $h\in\mathcal{T}_{\kappa}$,
we have 
\[
h\triangleright(f\star g)=(h_{(1)}\triangleright f)\star(h_{(2)}\triangleright g)\ ,\qquad(h\triangleright f)^{*}=S(h)^{*}\triangleright f^{*}\ .
\]

%% file: Hilbert.tex
Having at our disposal the algebra $\mathcal{A}$, the next step is
to introduce a Hilbert space together with a faithful $*$-representation
of the algebra on it. In this section we choose a weight $\omega$,
which is motivated by symmetry considerations, and use it to obtain
a Hilbert space $\mathcal{H}$ via the GNS construction. We show that
$\mathcal{H}$ is unitarily equivalent to $L^{2}(\mathbb{R}^{2})$
and determine the corresponding unitary operator $U$. Then we introduce
an unbounded $*$-representation $\rho$ of the extended momentum
algebra $\mathcal{T}_{\kappa}$ on $\mathcal{A}$, and prove that
$\rho(P_{\mu})$ and $\rho(\mathcal{E})$ are essentially self-adjoint
on $\mathcal{H}$. We show that the weight $\omega$ satisfies the
KMS condition with respect to a certain action $\alpha$, and determine
the corresponding modular operator $\Delta_{\omega}$. Finally we
prove some useful formulae that will be used extensively in the rest
of the paper.

\subsection{Definition of the Hilbert space $\mathcal{H}$}

We can introduce a Hilbert space, associated with the $*$-algebra
$\mathcal{A}$, via the GNS construction. The main ingredient is the
choice of a weight $\omega$ on $\mathcal{A}$ (recall that our algebra
is not unital). The inner product is defined in terms of $\omega$
by setting $(f,g):=\omega(f^{*}\star g)$ for $f,g\in\mathcal{A}$,
and the Hilbert space $\mathcal{H}$ is defined as the completion
of the algebra $\mathcal{A}$ in the norm induced by the inner product,
possibly after taking the quotient by the left ideal $\{f\in\mathcal{A}:\omega(f^{*}\star f)=0\}$.

A very natural choice for $\omega$ is given by the weight used in
the commutative case
\[
\omega(f):=\int f(x)d^{2}x\ ,
\]
where the integration is with respect to the Lebesgue measure. This
is motivated by simplicity, but more importantly by the invariance
of $\omega$ under the action of the $\kappa$-Poincaré algebra \cite{DuSi11}.
\begin{prop}
The weight $\omega$ is invariant under the action of $\mathcal{P}_{\kappa}$,
that is for any $f\in\mathcal{A}$ and for any $h\in\mathcal{P}_{\kappa}$
we have $\omega(h\triangleright f)=\varepsilon(h)\omega(f)$.
\end{prop}
We recall two additional results obtained in \cite{DuSi11}, which
we need in the following.
\begin{prop}
\label{prop:twisted-dusi}For any $f,g\in\mathcal{A}$ the weight
$\omega$ satisfies the twisted trace property
\[
\omega(f\star g)=\omega((\mathcal{E}\triangleright g)\star f)\ .
\]
Moreover we have
\[
\omega(f\star g^{*})=\int f(x)\overline{g}(x)d^{2}x\ .
\]

\end{prop}
Later the twisted trace property is going to be rewritten in the language
of KMS weights. Now we construct the Hilbert space using the GNS construction
for the weight $\omega$.
\begin{prop}
The inner product $(\cdot,\cdot)$ is non-degenerate. Completing $\mathcal{A}$
in the induced norm gives a Hilbert space $\mathcal{H}$, which is
unitarily equivalent to $L^{2}(\mathbb{R}^{2})$ via the unitary operator
\[
(Uf)(x)=\int e^{ip_{0}x_{0}}(\mathcal{F}_{0}f)(p_{0},e^{\lambda p_{0}}x_{1})\frac{dp_{0}}{2\pi}\ .
\]
\end{prop}
\begin{proof}
From the second part of Proposition \ref{prop:twisted-dusi} we have
\[
(f,g)=\omega(f^{*}\star g)=\int f^{*}(x)\overline{g^{*}}(x)d^{2}x\ .
\]
In particular if we set $f=g$ we have
\begin{equation}
\|f\|^{2}=(f,f)=\int|f^{*}(x)|^{2}d^{2}x\geq0\ .\label{eq:pos-norm}
\end{equation}
Notice that $f^{*}=0$ if and only if $f=0$, due to the properties
of the involution. Moreover we have $f^{*}\in C_{0}(\mathbb{R}^{2})$,
the functions vanishing at infinity on $\mathbb{R}^{2}$, and we know
that the Lebesgue integral is faithful on this space. Therefore we
have a positive-definite inner product and the left ideal $\{f\in\mathcal{A}:\omega(f^{*}\star f)=0\}$
consists only of $f=0$. By completing $\mathcal{A}$ with respect
to the norm $\|\cdot\|$ we obtain a Hilbert space, which we denote
by $\mathcal{H}$.

Now we prove that it is unitarily equivalent to $L^{2}(\mathbb{R}^{2})$
and determine the corresponding unitary operator $U:\mathcal{H}\to L^{2}(\mathbb{R}^{2})$.
Define $J\psi:=\psi^{*}$, where $*$ is the involution in $\mathcal{A}$,
and write $J:=J_{c}U$, where $J_{c}$ is complex conjugation. From
the formula (\ref{eq:alg-formulae}) we obtain
\[
(Uf)(x)=\int e^{ip_{0}x_{0}}(\mathcal{F}_{0}f)(p_{0},e^{\lambda p_{0}}x_{1})\frac{dp_{0}}{2\pi}\ .
\]
From equation (\ref{eq:pos-norm}) it follows that $U$ is an isometry
from $\mathcal{A}\subset\mathcal{H}$ to $L^{2}(\mathbb{R}^{2})$.
Indeed
\[
\|f\|^{2}=\int|(Jf)(x)|^{2}d^{2}x=\int|(Uf)(x)|^{2}d^{2}x=\|Uf\|_{L^{2}(\mathbb{R}^{2})}^{2}\ .
\]
Since $\mathcal{A}$ is a dense subset of $\mathcal{H}$ it follows
that $U$ can be extended by continuity to an isometry from $\mathcal{H}$
to $L^{2}(\mathbb{R}^{2})$. We still need to prove that $U$ is surjective.
We have that $U$ is invertible in $\mathcal{A}$ since $J=J_{c}U$
is invertible in $\mathcal{A}$, therefore the image of $U$ contains
$\mathcal{A}$. Notice that $\mathcal{A}$ is also a dense subset
of $L^{2}(\mathbb{R}^{2})$. Then, since $U$ is continuous, it contains
the closure of $\mathcal{A}$, which is indeed $L^{2}(\mathbb{R}^{2})$.
So we have shown that $U$ is isometric and surjective, therefore
unitary.\end{proof}
\begin{cor}
The involution, denoted here by $J$, is an antiunitary operator from
$\mathcal{H}$ to $L^{2}(\mathbb{R}^{2})$.\end{cor}
\begin{proof}
It follows immediately from the previous proposition. Indeed we have
that $U$ is unitary from $\mathcal{H}$ to $L^{2}(\mathbb{R}^{2})$
and complex conjugation is an antiunitary operator in $L^{2}(\mathbb{R}^{2})$
so, since $J=J_{c}U$, it follows that $J$ it is an antiunitary operator
from $\mathcal{H}$ to $L^{2}(\mathbb{R}^{2})$.
\end{proof}
For the construction of a spectral triple we need the following ingredient:
for any $f\in\mathcal{A}$ the $*$-representation $\pi(f)$ in $\mathcal{H}$,
given by the multiplication $\pi(f)\psi=f\star\psi$ for $\psi\in\mathcal{H}$,
should be a bounded operator. This follows from the next proposition.
\begin{prop}
For any $f\in\mathcal{A}$ the operator $\pi(f)$ is bounded on $\mathcal{H}$.\end{prop}
\begin{proof}
Recall that the $*$-algebra $\mathcal{A}$ is built from $L^{1}(G)$,
the convolution algebra of the Lie group $G$ associated to the $\kappa$-Minkowski
space in two dimensions. Any convolution algebra can be completed
to a $C^{*}$-algebra, either the group $C^{*}$-algebra $C^{*}(G)$
or the reduced one $C_{r}^{*}(G)$. In this case they coincide, since
$G$ is amenable, and we denote by $\|\cdot\|_{C^{*}(G)}$ the associated
$C^{*}$-norm. Now we can define a $C^{*}$-norm on $\mathcal{A}$
by setting $\|f\|:=\|\mathcal{F}f\|_{C^{*}(G)}$ for any $f\in\mathcal{A}$.
The verification that $\|\cdot\|$ is a $C^{*}$-norm is easy, indeed
using the formulae in (\ref{eq:alg-ops}) we have
\[
\begin{split}\|f\star f^{*}\| & =\|\mathcal{F}(f\star f^{*})\|_{C^{*}(G)}=\|\mathcal{F}f\hat{\star}\mathcal{F}f^{*}\|_{C^{*}(G)}\\
 & =\|\mathcal{F}f\hat{\star}(\mathcal{F}f)^{\hat{*}}\|_{C^{*}(G)}=\|\mathcal{F}f\|_{C^{*}(G)}^{2}=\|f\|^{2}\ .
\end{split}
\]
In the second line we have used the $C^{*}$-property of the norm
$\|\cdot\|_{C^{*}(G)}$. Then, since $\mathcal{H}$ is the Hilbert
space associated to $\mathcal{A}$ via the GNS construction, it follows
from the general properties of this construction that the operator
norm of $\pi(f)$ is bounded by $\|f\|$.
\end{proof}

\subsection{The representation of $\mathcal{T}_{\kappa}$}

In this subsection we are going to introduce an unbounded $*$-representation
$\rho$ of $\mathcal{T}_{\kappa}$ on a dense subspace of $\mathcal{H}$,
that is we extend the representation of the extended momentum algebra
to the Hilbert space. An obvious choice for the dense subspace is
$\mathcal{A}$, which by construction is dense in $\mathcal{H}$.
Recall that by Proposition \ref{prop:tk-module} we have that $\mathcal{A}$
is a left $\mathcal{T}_{\kappa}$-module $*$-algebra, where the representation
$\triangleright$ of $\mathcal{T}_{\kappa}$ on $\mathcal{A}$ is
defined by the formulae 
\begin{equation}
(P_{\mu}\triangleright f)(x)=-i(\partial_{\mu}f)(x)\ ,\qquad(\mathcal{E}\triangleright f)(x)=f(x_{0}+i\lambda,x_{1})\ .
\end{equation}
Then if we set $\rho(h)\psi:=h\triangleright\psi$, for every $h\in\mathcal{T}_{k}$
and $\psi\in\mathcal{A}$, we get an unbounded representation of $\mathcal{T}_{\kappa}$
on $\mathcal{A}$. Note that $\mathcal{A}$ is invariant under the
action of $\mathcal{T}_{\kappa}$ and the equivariance property for
the representation $\pi$ is automatic, since $\pi$ is given by left
multiplication and the equivariance property is just a restatement
of the fact that $\mathcal{A}$ is a left $\mathcal{T}_{\kappa}$-module
algebra.

To prove that $\rho$ is an unbounded $*$-representation we need
to show that the equality $(\rho(h)\phi,\psi)=(\phi,\rho(h^{*})\psi)$
holds for all $\phi,\psi\in\mathcal{A}$ and $h\in\mathcal{T}_{\kappa}$.
We only need to check this condition for the generators of $\mathcal{T}_{\kappa}$,
that is for the operators $\rho(P_{\mu})$ and $\rho(\mathcal{E})$
with domain $\mathcal{A}$. We are going to prove the stronger statement
that these operators are essentially self-adjoint on $\mathcal{H}$,
from which the previous equality follows. First we need a simple lemma.
\begin{lem}
For any $\psi\in\mathcal{A}$ we have $U\rho(P_{0})U^{-1}\psi=\rho(P_{0})\psi$
and $U\rho(P_{1})U^{-1}\psi=\rho(\mathcal{E})\rho(P_{1})\psi$.\end{lem}
\begin{proof}
Recall that, since $\mathcal{A}$ is a left $\mathcal{T}_{\kappa}$-module
$*$-algebra, we have the compatibility property with the involution
given by $(h\triangleright\psi)^{*}=S(h)^{*}\triangleright\psi^{*}$.
Here $S$ is the antipode map and the equality is valid for any $h\in\mathcal{T}_{\kappa}$
and $\psi\in\mathcal{A}$. In particular we have the equality $(P_{\mu}\triangleright\psi^{*})^{*}=S(P_{\mu})^{*}\triangleright\psi$.
Then, using the Hopf algebraic rules of $\mathcal{T}_{\kappa}$, one
immediately shows that $S(P_{0})^{*}=-P_{0}$ and $S(P_{1})^{*}=-\mathcal{E}^{-1}P_{1}$.
Now for any $\psi\in\mathcal{A}$ we have
\[
J\rho(P_{\mu})J\psi=(P_{\mu}\triangleright\psi^{*})^{*}=S(P_{\mu})^{*}\triangleright\psi\ .
\]
As a consequence we obtain the relations $J\rho(P_{0})J\psi=-\rho(P_{0})\psi$
and $J\rho(P_{1})J\psi=-\rho(\mathcal{E}^{-1})\rho(P_{1})\psi$. Using
the definition of the representation $\triangleright$ we obtain the
following formulae
\[
(J\rho(P_{0})J\psi)(x)=i(\partial_{0}\psi)(x)\ ,\qquad(J\rho(P_{1})J\psi)(x)=i(\partial_{1}\psi)(x_{0}-i\lambda,x_{1})\ .
\]
Now recall that we have the relations $J=J_{c}U$ and $J=U^{-1}J_{c}$,
where $J_{c}$ stands for complex conjugation. We can use the relations
to easily compute the following
\[
(U\rho(P_{1})U^{-1}\psi)(x)=(\overline{J\rho(P_{1})J\overline{\psi}})(x)=-i(\partial_{1}\psi)(x_{0}+i\lambda,x_{1})\ .
\]
Notice that this can be rewritten as $U\rho(P_{1})U^{-1}\psi=\rho(\mathcal{E})\rho(P_{1})\psi$.
In a similar way one shows that the identity $U\rho(P_{0})U^{-1}\psi=\rho(P_{0})\psi$
holds. The lemma is proven.
\end{proof}
Before starting the next proof we recall that $U$ is a unitary operator
from $\mathcal{H}$ to $L^{2}(\mathbb{R}^{2})$ and that $\mathcal{A}$
is dense in both Hilbert spaces.
\begin{prop}
The operators $\rho(P_{\mu})$ and $\rho(\mathcal{E})$, with domain
$\mathcal{A}$, are essentially self-adjoint on $\mathcal{H}$. Therefore
$\rho$ is an unbounded $*$-representation of $\mathcal{T}_{\kappa}$
on $\mathcal{A}$.\end{prop}
\begin{proof}
We can consider the operators $\rho(P_{0})$ and $\rho(\mathcal{E})\rho(P_{1})$
as unbounded operators on $L^{2}(\mathbb{R}^{2})$ with domain $\mathcal{A}$.
An easy computation shows that they are symmetric operators on this
Hilbert space. Using the previous lemma we can write for any $\phi,\psi\in\mathcal{A}$
\[
(\phi,\rho(\mathcal{E})\rho(P_{1})\psi)_{L^{2}(\mathbb{R}^{2})}=(\phi,U\rho(P_{1})U^{-1}\psi)_{L^{2}(\mathbb{R}^{2})}=(U^{-1}\phi,\rho(P_{1})U^{-1}\psi)\ .
\]
Now, since $\rho(\mathcal{E})\rho(P_{1})$ is symmetric on $L^{2}(\mathbb{R}^{2})$
and $U$ is invertible in $\mathcal{A}$, it follows that $\rho(P_{1})$
is symmetric in $\mathcal{H}$. The same argument applies to $\rho(P_{0})$.
Finally, it follows from Nelson's analytic vector theorem that $\rho(P_{0})$
and $\rho(P_{1})$ are essentially self-adjoint on $\mathcal{H}$,
since each element in $\mathcal{A}$ is an analytic vector for the
two symmetric operators. For $\rho(\mathcal{E})$ we only need to
observe that $\rho(\mathcal{E})=e^{-\lambda\rho(P_{0})}$ and apply
again Nelson's theorem.
\end{proof}
In the following we are going to denote the closure of these operators
by the same symbols, and we are also going to use the notation $\hat{P}_{\mu}:=\rho(P_{\mu})$.

\subsection{The KMS property of the weight $\omega$}

One relevant property of the weight $\omega$ is that is satisfies
the so-called twisted trace property, given in Proposition \ref{prop:twisted-dusi}.
This property can be recasted in the language of KMS weights \cite{kms-weights},
which provides great insight into the modular aspects of the spectral
triple we are constructing. Before proving this result we need to
review some results of Tomita-Takesaki modular theory. For simplicity
of presentation we restrict to the case of a faithful normal state,
but the following statements hold more generally for a faithful normal
semi-finite weight. The reader unfamiliar with these concepts is invited
to check the review \cite{modular-review}, which moreover contains
material on modular spectral triples, which are going to be of interest
to us later.

Given a von Neumann algebra $\mathcal{N}$ and a faithful normal state
$\omega$ on $\mathcal{N}$, the modular theory allows to create a
one-parameter group of $*$-automorphisms of the algebra $\mathcal{N}$,
which we call the \textit{modular automorphism group} associated to
$\omega$ and denote by $\sigma^{\omega}$. It is a map that assigns
to $t\in\mathbb{R}$ an automorphism of $\mathcal{N}$, which we denote
by $\sigma_{t}^{\omega}$. Denote by $\mathcal{H}_{\omega}$ the Hilbert
space constructed via the GNS construction for $\omega$, and denote
by $\pi_{\omega}$ the corresponding representation of $\mathcal{N}$
on $\mathcal{H}_{\omega}$. Then the modular automorphism group $\sigma^{\omega}$
is implemented by a unitary one-parameter group $t\mapsto\Delta_{\omega}^{it}\in\mathcal{B}(\mathcal{H}_{\omega})$.
This means that for each $a\in\mathcal{N}$ and for all $t\in\mathbb{R}$
we have $\pi_{\omega}(\sigma_{t}^{\omega}(a))=\Delta_{\omega}^{it}\pi_{\omega}(a)\Delta_{\omega}^{-it}$.
We call $\Delta_{\omega}$ the \textit{modular operator} associated
to the state $\omega$.

The modular automorphism group $\sigma^{\omega}$ has a very important
property: it is the unique one-parameter automorphism group that satisfies
the KMS condition with respect to the state $\omega$ at inverse temperature
$\beta=1$. The KMS condition is defined as follows.
\begin{defn}
Let $\mathcal{N}$ be a von Neumann algebra, $\omega$ a state on
$\mathcal{N}$ and $t\to\alpha_{t}$ a one-parameter group of automorphisms
of $\mathcal{N}$. Then $\omega$ satisfies the \textit{KMS condition}
at inverse temperature $\beta$ with respect to $\alpha$ if the following
conditions are satisfied:
\begin{enumerate}
\item for every $t\in\mathbb{R}$ we have $\omega\circ\alpha_{t}=\omega$,
\item for every $a,b\in\mathcal{N}$ there exists a bounded continuous function
$F_{a,b}$ from the horizontal strip $\{z\in\mathbb{C}:0\leq\mbox{Im}z\leq\beta\}$
to $\mathbb{C}$, which is analytic in the interior of the strip and
such that for every $t\in\mathbb{R}$ we have
\[
F_{a,b}(t)=\omega(a\alpha_{t}(b))\ ,\quad F_{a,b}(t+i\beta)=\omega(\alpha_{t}(b)a)\ .
\]

\end{enumerate}
\end{defn}
Now we are going to show how these concepts apply to our case. In
the next lemma we introduce a one-parameter group of $*$-automorphisms
of $\mathcal{A}$, which we denote by $\sigma^{\omega}$. This is
going to be the modular group of $\omega$, which justifies the notation.
\begin{lem}
For any $t\in\mathbb{R}$ and $f\in\mathcal{A}$ define $(\sigma_{t}^{\omega}f)(x):=f(x_{0}-\lambda t,x_{1})$.
We have that $\sigma^{\omega}$ is a one-parameter group of $*$-automorphisms
of $\mathcal{A}$.\end{lem}
\begin{proof}
To prove that $\sigma^{\omega}$ is one-parameter group of automorphisms
of $\mathcal{A}$ we have to show that, for any $f,g\in\mathcal{A}$
and $t\in\mathbb{R}$, the property $\sigma_{t}^{\omega}(f\star g)=\sigma_{t}^{\omega}(f)\star\sigma_{t}^{\omega}(g)$
is satisfied. This can be shown by a direct computation
\[
\begin{split}(\sigma_{t}^{\omega}(f)\star\sigma_{t}^{\omega}(g))(x) & =\int e^{ip_{0}x_{0}}(\mathcal{F}_{0}\sigma_{t}^{\omega}(f))(p_{0},x_{1})\sigma_{t}^{\omega}(g)(x_{0},e^{-\lambda p_{0}}x_{1})\frac{dp_{0}}{2\pi}\\
 & =\int e^{ip_{0}x_{0}}\int e^{-ip_{0}q_{0}}f(q_{0}-\lambda t,x_{1})g(x_{0}-\lambda t,e^{-\lambda p_{0}}x_{1})dq_{0}\frac{dp_{0}}{2\pi}\ .
\end{split}
\]
After the change of variable $q_{0}\to q_{0}+\lambda t$ we obtain
\[
\begin{split}(\sigma_{t}^{\omega}(f)\star\sigma_{t}^{\omega}(g))(x) & =\int e^{ip_{0}(x_{0}-\lambda t)}\int e^{-ip_{0}q_{0}}f(q_{0},x_{1})g(x_{0}-\lambda t,e^{-\lambda p_{0}}x_{1})dq_{0}\frac{dp_{0}}{2\pi}\\
 & =\int e^{ip_{0}(x_{0}-\lambda t)}(\mathcal{F}_{0}f)(p_{0},x_{1})g(x_{0}-\lambda t,e^{-\lambda p_{0}}x_{1})\frac{dp_{0}}{2\pi}=(\sigma_{t}^{\omega}(f\star g))(x)\ .
\end{split}
\]
Finally to prove that $\sigma_{t}^{\omega}$ is a $*$-automorphism
we need to check the additional property $\sigma_{t}^{\omega}(f)^{*}=\sigma_{t}^{\omega}(f^{*})$.
This can again be checked by a direct computation
\[
\begin{split}\sigma_{t}^{\omega}(f)^{*}(x) & =\int e^{ip_{0}x_{0}}(\mathcal{F}_{0}\overline{\alpha_{t}(f)})(p_{0},e^{-\lambda p_{0}}x_{1})\frac{dp_{0}}{2\pi}\\
 & =\int e^{ip_{0}x_{0}}\int e^{-ip_{0}q_{0}}\overline{f}(q_{0}-\lambda t,e^{-\lambda p_{0}}x_{1})dq_{0}\frac{dp_{0}}{2\pi}
\end{split}
\]
Using again the change of variable $q_{0}\to q_{0}+\lambda t$ we
obtain
\[
\begin{split}\sigma_{t}^{\omega}(f)^{*}(x) & =\int e^{ip_{0}(x_{0}-\lambda t)}\int e^{-ip_{0}q_{0}}\overline{f}(q_{0},e^{-\lambda p_{0}}x_{1})dq_{0}\frac{dp_{0}}{2\pi}\\
 & =\int e^{ip_{0}(x_{0}-\lambda t)}(\mathcal{F}_{0}\overline{f})(p_{0},e^{-\lambda p_{0}}x_{1})\frac{dp_{0}}{2\pi}=\sigma_{t}^{\omega}(f^{*})(x)\ .
\end{split}
\]
\end{proof}
\begin{prop}
The weight $\omega$ satisfies the KMS condition at inverse temperature
$\beta=1$ with respect to $\sigma^{\omega}$. The corresponding modular
operator is given by $\Delta_{\omega}=e^{-\lambda\hat{P}_{0}}$.\end{prop}
\begin{proof}
We define the function $F_{f,g}(z):=\omega(f\star\sigma_{z}^{\omega}g)$.
It is bounded continuous and analytic, since $\sigma^{\omega}$ acts
on the first variable and functions in $\mathcal{A}$ are analytic
in the first variable. To prove that $\omega$ satisfied the KMS condition
with respect to $\sigma^{\omega}$ we need to show that
\[
F_{f,g}(t)=\omega(f\star\sigma_{t}^{\omega}(g))\ ,\quad F_{f,g}(t+i\beta)=\omega(\sigma_{t}^{\omega}(g)\star f)\ .
\]
Notice that the action of $\mathcal{E}$ can be rewritten in terms
of $\sigma^{\omega}$, that is 
\begin{equation}
(\mathcal{E}\triangleright f)(x)=f(x_{0}+i\lambda,x_{1})=(\sigma_{-i}^{\omega}f)(x)\ .
\end{equation}
Then using the twisted trace property we have
\[
\begin{split}F_{f,g}(t+i) & =\omega(f\star\sigma_{t+i}^{\omega}(g))=\omega((\mathcal{E}\triangleright\sigma_{t+i}^{\omega}(g))\star f)\\
 & =\omega((\sigma_{-i}^{\omega}\sigma_{t+i}^{\omega}(g))\star f)=\omega(\sigma_{t}^{\omega}(g)\star f)\ .
\end{split}
\]
This proves the KMS condition. To determine the modular operator $\Delta_{\omega}$
associated with $\omega$ consider $f,g\in\mathcal{A}$. Using the
fact that $\sigma_{t}^{\omega}$ is an automorphism of $\mathcal{A}$
we have
\begin{equation}
\pi(\sigma_{t}^{\omega}(f))g=\sigma_{t}^{\omega}(f)\star g=\sigma_{t}^{\omega}(f\star\sigma_{-t}^{\omega}(g))=\sigma_{t}^{\omega}(\pi(f)\sigma_{-t}^{\omega}(g))\ .\label{eq:aut-modul}
\end{equation}
Now, using the fact that $\hat{P}_{0}=-i\partial_{0}$, we have the
following equality
\[
(\sigma_{t}^{\omega}f)(x)=f(x_{0}-\lambda t,x_{1})=(e^{-i\lambda t\hat{P}_{0}}f)(x)\ .
\]
Then we see that equation (\ref{eq:aut-modul}) can be rewritten as
\[
\pi(\sigma_{t}^{\omega}(f))g=e^{-i\lambda t\hat{P}_{0}}\pi(f)e^{i\lambda t\hat{P}_{0}}g=\Delta_{\omega}^{it}\pi(f)\Delta_{\omega}^{-it}g\ .
\]
This implies that the modular operator is given by $\Delta_{\omega}=e^{-\lambda\hat{P}_{0}}$.
\end{proof}

\subsection{Some useful formulae}

We have seen that the Hilbert space $\mathcal{H}$, constructed via
the GNS construction for $\omega$, is unitarily equivalent to $L^{2}(\mathbb{R}^{2})$,
where the unitary operator is given by
\[
(Uf)(x)=\int e^{ip_{0}x_{0}}(\mathcal{F}_{0}f)(p_{0},e^{\lambda p_{0}}x_{1})\frac{dp_{0}}{2\pi}\ .
\]
We can associate, to any densely defined operator $T$ on $\mathcal{H}$,
the densely defined operator $UTU^{-1}$ on $L^{2}(\mathbb{R}^{2})$.
Many properties of operators defined on $\mathcal{H}$ are conserved
by this unitary transformation: for example, if $T$ belongs to the
$p$-th Schatten ideal on $\mathcal{H}$, then $UTU^{-1}$ belongs
to the $p$-th Schatten ideal on $L^{2}(\mathbb{R}^{2})$. This is
useful, since we can use results which are formulated for the Hilbert
space $L^{2}(\mathbb{R}^{2})$ to establish some properties of operators
on $\mathcal{H}$. For example, to prove that an operator is Hilbert-Schmidt
in $L^{2}(\mathbb{R}^{2})$, one can write it in the form of an integral
operator and check that the kernel belongs to $L^{2}(\mathbb{R}^{2}\times\mathbb{R}^{2})$.

Here we collect some useful formulae used extensively in the rest
of the paper.
\begin{lem}
\label{lem:mult-formula}For $f\in\mathcal{A}$ and $\psi\in\mathcal{H}$
we have
\[
(U\pi(f)U^{-1}\psi)(x)=\int e^{ipx}(Uf)(x_{0},e^{\lambda p_{0}}x_{1})(\mathcal{F}\psi)(p)\frac{d^{2}p}{(2\pi)^{2}}\ .
\]
\end{lem}
\begin{proof}
Consider first $\psi\in\mathcal{A}$. From the definition of $\pi(f)$
and the properties of the involution $J$ we obtain that $J\pi(f)J$
acts by multiplication from the right, that is we have $J\pi(f)J\psi=\psi\star Jf$.
To compute $U\pi(f)U^{-1}$ we use the fact that $U$ is related to
the involution by $J=J_{c}U$, where $J_{c}$ denotes complex conjugation.
We have the following equality
\[
U\pi(f)U^{-1}\psi=J_{c}J\pi(f)JJ_{c}\psi=J_{c}(J_{c}\psi\star Jf)=J_{c}(J_{c}\psi\star J_{c}Uf)\ .
\]
Using the formula (\ref{eq:alg-formulae}) for the product and after
a change of variable we obtain
\[
\begin{split}J_{c}(J_{c}\psi\star J_{c}Uf)(x) & =\int e^{-ip_{0}x_{0}}(\overline{\mathcal{F}_{0}\overline{\psi}})(p_{0},x_{1})(Uf)(x_{0},e^{-\lambda p_{0}}x_{1})\frac{dp_{0}}{2\pi}\\
 & =\int e^{ip_{0}x_{0}}(\mathcal{F}_{0}\psi)(p_{0},x_{1})(Uf)(x_{0},e^{\lambda p_{0}}x_{1})\frac{dp_{0}}{2\pi}\ .
\end{split}
\]
Finally this expression may be rewritten as
\[
(U\pi(f)U^{-1}\psi)(x)=\int e^{ipx}(Uf)(x_{0},e^{\lambda p_{0}}x_{1})(\mathcal{F}\psi)(p)\frac{d^{2}p}{(2\pi)^{2}}\ .
\]
This formula extends by continuity for any $\psi\in L^{2}(\mathbb{R}^{2})$.
Indeed since $\pi(f)$ is a bounded operator in $\mathcal{H}$ and
since $U$ is a unitary operator from $\mathcal{H}$ to $L^{2}(\mathbb{R}^{2})$
it follows that $U\pi(f)U^{-1}$ is bounded in $L^{2}(\mathbb{R}^{2})$.
Therefore $U\pi(f)U^{-1}$ may be extended by continuity to $L^{2}(\mathbb{R}^{2})$.
\end{proof}
The following lemma gives an explicit expression for a function of
$\hat{P}_{\mu}$.
\begin{lem}
\label{lem:g-formula}For $g$ a bounded function, define $g(\hat{P})$
by the functional calculus. We have 
\[
(Ug(\hat{P})U^{-1}\psi)(x)=\int e^{ipx}g(p_{0},e^{-\lambda p_{0}}p_{1})(\mathcal{F}\psi)(p)\frac{d^{2}p}{(2\pi)^{2}}\ .
\]
\end{lem}
\begin{proof}
We know that in $L^{2}(\mathbb{R}^{2})$ the Fourier transform $\mathcal{F}$
is the unitary operator that diagonalizes the operators $\tilde{P}_{\mu}=-i\partial_{\mu}$
(up to factors of $2\pi$, depending on the normalization). This means
that for $\psi\in L^{2}(\mathbb{R}^{2})$ we have
\[
(g(\tilde{P})\psi)(x)=\int e^{ipx}g(p)(\mathcal{F}\psi)(p)\frac{d^{2}p}{(2\pi)^{2}}\ .
\]
From the results of the previous section we have that $U\hat{P}_{0}U^{-1}=\tilde{P}_{0}$
and $U\hat{P}_{1}U^{-1}=e^{-\lambda\tilde{P}_{0}}\tilde{P}_{1}$,
so we obtain $Ug(\hat{P})U^{-1}=g(\tilde{P}_{0},e^{-\lambda\tilde{P}_{0}}\tilde{P}_{1})$.
The result follows by using the functional calculus for the commuting
operators $\tilde{P}_{\mu}$ in $L^{2}(\mathbb{R}^{2})$.
\end{proof}
Finally we are interested in operators of the form $\pi(f)g(\hat{P})$.
\begin{prop}
\label{prop:kernel}The Schwartz kernel associated to the operator
$U\pi(f)g(\hat{P})U^{-1}$ is given by
\[
K(x,y)=\int e^{ip(x-y)}(Uf)(x_{0},e^{\lambda p_{0}}x_{1})g(p_{0},e^{-\lambda p_{0}}p_{1})\frac{d^{2}p}{(2\pi)^{2}}\ .
\]
\end{prop}
\begin{proof}
Using Lemmata \ref{lem:mult-formula} and \ref{lem:g-formula} we
obtain
\[
\begin{split}(U\pi(f)g(\hat{P})U^{-1}\psi)(x) & =(U\pi(f)U^{-1}Ug(\hat{P})U^{-1}\psi)(x)\\
 & =\int e^{ipx}(Uf)(x_{0},e^{\lambda p_{0}}x_{1})(\mathcal{F}Ug(\hat{P})U^{-1}\psi)(p)\frac{d^{2}p}{(2\pi)^{2}}\\
 & =\int e^{ipx}(Uf)(x_{0},e^{\lambda p_{0}}x_{1})g(p_{0},e^{-\lambda p_{0}}p_{1})(\mathcal{F}\psi)(p)\frac{d^{2}p}{(2\pi)^{2}}\ .
\end{split}
\]
The Schwartz kernel corresponding to this operator is given by
\[
K(x,y)=\int e^{ip(x-y)}(Uf)(x_{0},e^{\lambda p_{0}}x_{1})g(p_{0},e^{-\lambda p_{0}}p_{1})\frac{d^{2}p}{(2\pi)^{2}}\ .
\]
\end{proof}

%% file: Dirac.tex
The next step in the construction of a spectral triple is the introduction
of a self-adjoint unbounded operator $\mathcal{D}$ on the Hilbert
space, which we call the Dirac operator. In the next section we briefly
recall the ingredients used in the commutative case, mainly to fix
some notation. We show that there is a problem in obtaining a bounded
commutator, which is solved by considering the twisted commutator
introduced by Connes and Moscovici \cite{type III}, where the commutator
is twisted by an automorphism $\sigma$. Finally we prove that, under
some assumptions related to symmetry and to the classical limit, there
is a unique pair of a Dirac operator $\mathcal{D}$ and automorphism
$\sigma$ such that the twisted commutator is bounded.

\subsection{A problem with boundedness}

Let us briefly recall some facts concerning the Dirac operator $\mathcal{D}$
in the two dimensional Euclidean space $\mathbb{R}^{2}$. Here we
consider the algebra of Schwartz functions $\mathcal{A}=\mathcal{S}(\mathbb{R}^{2})$
with the $*$-representation $\pi$ on $\mathcal{H}_{r}=L^{2}(\mathbb{R}^{2})$
given by pointwise multiplication, that is $(\pi(f)\psi)(x)=f(x)\psi(x)$.
We consider the following representation of the Clifford algebra 
\[
\Gamma^{0}:=\sigma^{1}=\left(\begin{array}{cc}
0 & 1\\
1 & 0
\end{array}\right)\ ,\qquad\Gamma^{1}:=\sigma^{2}=\left(\begin{array}{cc}
0 & -i\\
i & 0
\end{array}\right)\ .
\]
These matrices satisfy the anticommutation relations $\{\Gamma^{\mu},\Gamma^{\nu}\}=2\delta^{\mu\nu}$.
To consider spinors we define the Hilbert space $\mathcal{H}=\mathcal{H}_{r}\otimes\mathbb{C}^{2}$,
corresponding to the trivial spinor bundle. The representation $\pi$
of $\mathcal{A}$ is extended as pointwise multiplication on the two
copies of $\mathcal{H}_{r}$, and we denote it again by $\pi$. The
inner product is given by
\[
(\psi,\phi)_{\mathcal{H}}=\int\left(\overline{\psi_{1}}(x)\phi_{1}(x)+\overline{\psi_{2}}(x)\phi_{2}(x)\right)d^{2}x\ ,
\]
where $\psi_{i},\phi_{j}$ are the components of the spinors $\psi,\phi$.
The Dirac operator for this space is built using the $\Gamma$ matrices
and the operators $\hat{P}_{\mu}=-i\partial_{\mu}$ as follows 
\[
\mathcal{D}=\Gamma^{\mu}\hat{P}_{\mu}=-\left(\begin{array}{cc}
0 & i\partial_{0}+\partial_{1}\\
i\partial_{0}-\partial_{1} & 0
\end{array}\right)\ .
\]
The Dirac operator $\mathcal{D}$ is essentially self-adjoint with
respect to the inner product defined above. Since the dimension is
even there is a self-adjoint operator $\chi$, called the grading,
which satisfies $\chi^{2}=1$ and the following properties: it commutes
with the algebra, that is for any $f\in\mathcal{A}$ we have $[\chi,\pi(f)]=0$,
and it anticommutes with the Dirac operator, that is $\{\chi,\mathcal{D}\}=0$.
In terms of the $\Gamma$ matrices it is given by
\[
\chi=-i\Gamma^{0}\Gamma^{1}=\left(\begin{array}{cc}
1 & 0\\
0 & -1
\end{array}\right)\ .
\]
One of the requirements in the definition of a spectral triple is
that $[\mathcal{D},\pi(f)]$, for $f\in\mathcal{A}$, should extend
to a bounded operator in $\mathcal{H}$. From its definition we obtain
$[\mathcal{D},\pi(f)]=\Gamma^{\mu}[\hat{P}_{\mu},\pi(f)]$. Using
the Leibnitz rule for the derivatives we obtain
\[
\hat{P}_{\mu}\pi(f)\psi=(\hat{P}_{\mu}f)\psi+f(\hat{P}_{\mu}\psi)=\pi(\hat{P}_{\mu}f)\psi+\pi(f)\hat{P}_{\mu}\psi\ .
\]
The previous equation can be rewritten in the form $[\mathcal{D},\pi(f)]=\Gamma^{\mu}\pi(\hat{P}_{\mu}f)$.
This is a bounded operator, since $\pi(f)$ is bounded for $f\in\mathcal{A}$
and $\hat{P}_{\mu}f\in\mathcal{A}$, since $f$ is a Schwartz function.

Now let us discuss the deformed case. We refer to the Hilbert space
introduced in the previous section as the reduced Hilbert space, and
we denote it by $\mathcal{H}_{r}$. The algebra $\mathcal{A}$ is
represented on $\mathcal{H}_{r}$ by left multiplication, that is
$(\pi(f)\psi)(x)=(f\star\psi)(x)$. To consider spinors we define
the Hilbert space $\mathcal{H}=\mathcal{H}_{r}\otimes\mathbb{C}^{2}$,
corresponding to the trivial spinor bundle. The representation $\pi$
of $\mathcal{A}$ is trivially extended to the two copies of $\mathcal{H}_{r}$,
and we denote it again by $\pi$. The inner product of two spinors
$\psi,\phi\in\mathcal{H}$ is given by
\[
(\psi,\phi)_{\mathcal{H}}=\int\left((\psi_{1}^{*}\star\phi_{1})(x)+(\psi_{2}^{*}\star\phi_{2})(x)\right)d^{2}x\ .
\]

As a first attempt, we can try to use the classical Dirac operator
$\mathcal{D}=\Gamma^{\mu}\hat{P}_{\mu}$. Then, as in the commutative
case, we have $[\mathcal{D},\pi(f)]=\Gamma^{\mu}[\hat{P}_{\mu},\pi(f)]$.
The difference with the commutative case comes from the fact that
the representation $\pi$ is not pointwise multiplication. Indeed,
using the fact that the representation $\pi$ is $\mathcal{T}_{\kappa}$-equivariant,
we obtain
\[
\hat{P}_{1}\pi(f)\psi=\rho(P_{1})\pi(f)\psi=\pi(P_{1}\triangleright f)\psi+\pi(\mathcal{E}\triangleright f)\rho(P_{1})\psi\ .
\]
As a consequence of the non-trivial coproduct of the $\kappa$-Poincaré
algebra, $P_{1}$ does not obey the Leibnitz rule. Then the commutator
$[\mathcal{D},\pi(f)]$ is not bounded, since $\rho(P_{1})$ is an
unbounded operator. Explicitely we have 
\[
[\mathcal{D},\pi(f)]=\Gamma^{\mu}\pi(P_{\mu}\triangleright f)+\Gamma^{1}\pi((\mathcal{E}-1)\triangleright f)\rho(P_{1})\ .
\]
One can hope to evade this problem by considering a different Dirac
operator $\mathcal{D}$. To have a good classical limit we require
that, in some technical sense that we will specify later, in the limit
$\lambda\to0$ this Dirac operator $\mathcal{D}$ reduces to the classical
one. Moreover, since we are dealing with a quantum group, it is natural
to consider the framework of twisted spectral triples \cite{type III},
see the discussions in \cite{wagner,quantum-twisted}. Then we require
the boundedness of 
\[
[\mathcal{D},\pi(f)]_{\sigma}=\mathcal{D}\pi(f)-\pi(\sigma(f))\mathcal{D}\ .
\]
Here $\sigma$ is an automorphism of the algebra $\mathcal{A}$. Since
the algebra $\mathcal{A}$ is involutive one also requires the compatibility
property $\sigma(f)^{*}=\sigma^{-1}(f^{*})$, for $f\in\mathcal{A}$.
In the next subsection we investigate what are the possible choices
for the Dirac operator $\mathcal{D}$ and the automorphism $\sigma$.

\subsection{The deformed Dirac operator $\mathcal{D}$}

Now we state our assumptions for the Dirac operator $\mathcal{D}$
and the automorphism $\sigma$. The first assumption is of general
nature, that is we require that $\mathcal{D}$ is self-adjoint on
$\mathcal{H}=\mathcal{H}_{r}\otimes\mathbb{C}^{2}$ and that it anticommutes
with the grading $\chi$. This implies that $\mathcal{D}$ is of the
form $\mathcal{D}=\Gamma^{\mu}\hat{D}_{\mu}$, where $\hat{D}_{\mu}$
are self-adjoint operators on $\mathcal{H}_{r}$. The second assumption
is that, in the limit $\lambda\to0$, $\mathcal{D}$ and $\sigma$
should reduce respectively to the classical Dirac operator and the
identity. The technical statement of this assumption is going to be
given later. Next, motivated by the fact that $\mathcal{D}$ and $\sigma$
should be determined by the symmetries, we assume that $\hat{D}_{\mu}=\rho(D_{\mu})$,
for some $D_{\mu}\in\mathcal{T}_{\kappa}$, and that the automorphism
$\sigma$ is given by $\sigma(f)=\sigma\triangleright f$, for some
$\sigma\in\mathcal{T}_{\kappa}$. The requirement that $\sigma$ is
an automorphism implies that its coproduct is $\Delta(\sigma)=\sigma\otimes\sigma$. 

We recall that we do not consider any topology on $\mathcal{T}_{\kappa}$,
therefore any element can be written as a \textit{finite }sum of products
of generators. Since the algebra is commutative, any such element
is a sum of terms of the form $P_{0}^{i}P_{1}^{j}\mathcal{E}^{k}$,
with $i,j\in\mathbb{N}$ and $k\in\mathbb{Z}$. Notice that, in our
definition of $\mathcal{T}_{\kappa}$, the element $\mathcal{E}$
is not the formal series $e^{-\lambda P_{0}}$ in $P_{0}$, as usually
given in the defining relations of $\kappa$-Poincaré, but is considered
as one of the generators of the algebra.
\begin{lem}
\label{lem:twist-bnd}Suppose that $A,\sigma\in\mathcal{T}_{\kappa}$.
Then $[\rho(A),\pi(f)]_{\sigma}=\rho(A)\pi(f)-\pi(\sigma\triangleright f)\rho(A)$
is bounded if and only if the coproduct of $A$ is $\Delta(A)=A^{\prime}\otimes1+\sigma\otimes A$,
for some $A^{\prime}\in\mathcal{T}_{\kappa}$.\end{lem}
\begin{proof}
Using the fact that the representation $\pi$ is equivariant we have
\[
\begin{split}[\rho(A),\pi(f)]_{\sigma} & =\rho(A)\pi(f)-\pi(\sigma\triangleright f)\rho(A)\\
 & =\pi(A_{(1)}\triangleright f)\rho(A_{(2)})-\pi(\sigma\triangleright f)\rho(A)\ .
\end{split}
\]
For some $B_{i},C_{j}\in\mathcal{T}_{\kappa}$ this can be rewritten
in the form 
\begin{equation}
[\rho(A),\pi(f)]_{\sigma}=\sum_{ij}\pi(B_{i}\triangleright f)\rho(C_{j})\ .\label{eq:hopf-comm}
\end{equation}
Now we show that the operator $\rho(C_{j})$ is bounded if and only
if $C_{j}$ is a multiple of the unit. Indeed, from their definition,
we have that $\rho(P_{\mu})$ and $\rho(\mathcal{E})$ are unbounded
operators, while the unit corresponds to the identity operator. Using
our previous remark on the structure of $\mathcal{T}_{\kappa}$, we
have that a general operator $\rho(h)$, with $h\in\mathcal{T}_{\kappa}$,
can be written in the form 
\[
\rho(h)=\sum_{ijk}c_{ijk}\rho(P_{0}^{i})\rho(P_{1}^{j})\rho(\mathcal{E}^{k})\ .
\]
This is unbounded unless all its coefficients are zero, except possibly
for $c_{000}$. Then it follows from equation (\ref{eq:hopf-comm})
that $[\rho(A),\pi(f)]_{\sigma}$ is bounded if and only if all the
elements $C_{j}$ are multiples of the unit. Setting $C_{j}=c_{j}1$,
with $c_{j}\in\mathbb{C}$, we can rewrite equation (\ref{eq:hopf-comm})
as
\[
[\rho(A),\pi(f)]_{\sigma}=\pi(A^{\prime}\triangleright f)\ ,\qquad A^{\prime}:=\sum_{ij}c_{j}B_{i}\ .
\]
Using the definition of the twisted commutator, we have that $[\rho(A),\pi(f)]_{\sigma}=\pi(A^{\prime}\triangleright f)$
implies
\[
\rho(A)\pi(f)=\pi(A^{\prime}\triangleright f)+\pi(\sigma\triangleright f)\rho(A)\ .
\]
Finally, using the equivariance of $\pi$, this implies that $\Delta(A)=A^{\prime}\otimes1+\sigma\otimes A$.
\end{proof}
The previous lemma tells us that the requirement of boundedness of
the twisted commutator severly restricts the form of the coproduct
of $D_{\mu}$. Now we characterize which elements of $\mathcal{T}_{\kappa}$
have a coproduct of this form.
\begin{lem}
\label{lem:list-elems}Consider an element $A\in\mathcal{T}_{\kappa}$.
Suppose that $\Delta(A)=A^{\prime}\otimes1+\sigma\otimes A$ for some
$A^{\prime},\sigma\in\mathcal{T}_{\kappa}$, and that $\Delta(\sigma)=\sigma\otimes\sigma$.
Then we have $\sigma=\mathcal{E}^{m}$, for some $m\in\mathbb{Z}$,
and for some coefficients $c_{j}\in\mathbb{C}$ we have that $A$
can be written as
\begin{itemize}
\item $c_{1}1+c_{2}\mathcal{E}^{m}$ if $m<0$,
\item \textup{$c_{1}1+c_{2}P_{0}$} if $m=0$,
\item $c_{1}1+c_{2}\mathcal{E}+c_{3}P_{1}$ if $m=1$,
\item $c_{1}1+c_{2}\mathcal{E}^{m}$ if $m>1$.
\end{itemize}
\end{lem}
\begin{proof}
First of all it is clear that if $\sigma\in\mathcal{T}_{\kappa}$
is such that $\Delta(\sigma)=\sigma\otimes\sigma$, then we must have
$\sigma=\mathcal{E}^{m}$ for some $m\in\mathbb{Z}$. Any element
$A\in\mathcal{T}_{\kappa}$ is a finite sum of elements of the form
$P_{0}^{i}P_{1}^{j}\mathcal{E}^{k}$, with $i,j\in\mathbb{N}$ and
$k\in\mathbb{Z}$, and we can write it as
\[
A=\sum_{ijk}c_{ijk}P_{0}^{i}P_{1}^{j}\mathcal{E}^{k}=\sum_{ijk}A_{ijk}\ .
\]
The request $\Delta(A)=A^{\prime}\otimes1+\sigma\otimes A$ puts a
constraint on the possible terms that appear in the sum. A term $A_{ijk}$
is allowed if and only if its coproduct is of the form $\Delta(A_{ijk})=B_{ijk}\otimes1+\sigma\otimes C_{ijk}$,
for some $B_{ijk},C_{ijk}\in\mathcal{T}_{\kappa}$. Now we discuss
which terms are of this form.

Consider first the generator $P_{0}$. Then $P_{0}^{i}$ is allowed
for $i=0,1$, since we have $\Delta(1)=1\otimes1$ and $\Delta(P_{0})=P_{0}\otimes1+1\otimes P_{0}$.
For higher powers we get cross terms which cannot be of the allowed
form, for example if we compute the coproduct for $i=2$ we find
\[
\Delta(P_{0}^{2})=P_{0}^{2}\otimes1+2P_{0}\otimes P_{0}+1\otimes P_{0}^{2}\ .
\]
The same argument applies to $P_{1}$, for which we have $\Delta(P_{1})=P_{1}\otimes1+\mathcal{E}\otimes P_{1}$.
Higher powers are not allowed as in the case of $P_{0}$. For $\mathcal{E}$
on the other hand any power is acceptable, since $\Delta(\mathcal{E}^{k})=\mathcal{E}^{k}\otimes\mathcal{E}^{k}$
is an automorphism. Now we consider the mixed terms. Any term of the
form $P_{0}^{i}P_{1}^{j}$ for $i,j\geq1$ is not allowed because
of the cross terms. For example for $P_{0}P_{1}$ we have
\[
\Delta(P_{0}P_{1})=P_{0}P_{1}\otimes1+\mathcal{E}P_{0}\otimes P_{1}+P_{1}\otimes P_{0}+\mathcal{E}\otimes P_{0}P_{1}\ .
\]
Similarly mixed terms like $P_{0}^{i}\mathcal{E}^{k}$ or $P_{1}^{j}\mathcal{E}^{k}$
are not allowed for $i,j,k\geq1$. For example we have
\[
\begin{split}\Delta(P_{0}\mathcal{E}^{k}) & =\mathcal{E}^{k}P_{0}\otimes\mathcal{E}^{k}+\mathcal{E}^{k}\otimes\mathcal{E}^{k}P_{0}\ ,\\
\Delta(P_{1}\mathcal{E}^{k}) & =\mathcal{E}^{k}P_{1}\otimes\mathcal{E}^{k}+\mathcal{E}^{k}\otimes\mathcal{E}^{k}P_{1}\ .
\end{split}
\]

Now we discuss which terms are compatible when we fix $\sigma=\mathcal{E}^{m}$,
for some $m\in\mathbb{Z}$. The coproduct of $A_{ijk}$ must have
the form $\Delta(A_{ijk})=B_{ijk}\otimes1+\mathcal{E}^{m}\otimes C_{ijk}$.
The unit and $\mathcal{E}^{m}$ satisfy this requirement for any $m\in\mathbb{Z}$.
They are the only possible terms for $m<0$ and $m>1$. On the other
hand for $m=0$ we can also have $P_{0}$, while for $m=1$ we can
also have $P_{1}$.
\end{proof}
Now it is time to discuss the requirement of the classical limit for
$\mathcal{D}$.\textit{ }First of all we need to recall that the parameter
$\lambda$ is a \textit{physical quantity} of the model, which has
the physical dimension of a length. Since also the coordinates $x_{\mu}$
have the dimensions of a length, it follows that the Dirac operator
must have dimension $[\mathcal{D}]=-1$, where by $[A]$ we denote
the physical dimension of $A$ in units of length. From this observation
it follows that the generators of $\mathcal{T}_{\kappa}$ have physical
dimensions $[P_{\mu}]=-1$, $[\mathcal{E}]=0$ and the unit has $[1]=0$.
Now we can give the technical statement of the assumption of the classical
limit.
\begin{defn}
We say that the Dirac operator $\mathcal{D}=\Gamma^{\mu}\hat{D}_{\mu}$
obeys the \textit{classical limit} if for all $\psi\in\mathcal{A}$
we have $\lim\hat{D}_{\mu}\psi=\hat{P}_{\mu}\psi$ for $\lambda\to0$,
and moreover $[\hat{D}_{\mu}]=-1$.
\end{defn}
Now that we have stated all the assumptions we can prove the following
theorem.
\begin{thm}
Suppose that $D_{\mu},\sigma\in\mathcal{T}_{\kappa}$ and that the
Dirac operator $\mathcal{D}=\Gamma^{\mu}\rho(D_{\mu})$ obeys the
classical limit. Then the twisted commutator
\[
[\mathcal{D},\pi(f)]_{\sigma}=\Gamma^{\mu}\left(\rho(D_{\mu})\pi(f)-\pi(\sigma\triangleright f)\rho(D_{\mu})\right)
\]
is bounded if and only if we have $D_{0}=\frac{1}{\lambda}(1-\mathcal{E})$,
$D_{1}=P_{1}$ and $\sigma=\mathcal{E}$.\end{thm}
\begin{proof}
Using Lemma \ref{lem:twist-bnd} we have that $[\mathcal{D},\pi(f)]_{\sigma}$
is bounded if and only if $\Delta(D_{\mu})=D_{\mu}^{\prime}\otimes1+\sigma\otimes D_{\mu}$,
for some $D_{\mu}^{\prime}\in\mathcal{T}_{\kappa}$. Such elements
are classified by Lemma \ref{lem:list-elems}, depending on $m\in\mathbb{Z}$.
The condition that the operators $\hat{D}_{\mu}$ obey the classical
limit imposes $m=1$. Indeed only for this choice we have the element
$P_{1}$, which corresponds to the operator $\rho(P_{1})=-i\partial_{1}$.

Now we need to consider the restriction on the coefficients, imposed
again by the classical limit, for $m=1$. Let us start with $D_{0}$,
for which we have
\[
(\rho(D_{0})\psi)(x)=\lambda^{-1}\left(c_{1}\psi(x)+c_{2}\psi(x_{0}-i\lambda,x_{1})\right)-ic_{3}(\partial_{1}\psi)(x)\ .
\]
We have rescaled the coefficients in such a way that $c_{k}\in\mathbb{C}$
are dimensionless. The operator $\rho(D_{0})$ should reduce to $-i\partial_{0}$
in the limit $\lambda\to0$. We see immediately that we must have
$c_{3}=0$, since this term is not affected by the limit. To see what
are the requirements on the remaining two coefficients, we expand
$\psi$ in Taylor series in $\lambda$, which is allowed since $\psi\in\mathcal{A}$
is analytic in the first variable. We have
\[
(\rho(D_{0})\psi)(x)=\lambda^{-1}(c_{1}+c_{2})\psi(x)+ic_{2}(\partial_{0}\psi)(x)+\lambda^{-1}c_{2}\sum_{n=2}^{\infty}(\partial_{0}^{n}\psi)(x)\frac{(i\lambda)^{n}}{n!}\ .
\]
The first term diverges in the limit $\lambda\to0$ unless $c_{1}=-c_{2}$,
which we must require. The second term on the other hand is not affected
by this limit, and forces us to choose $c_{2}=-1$. For the third
term, after exchanging the limit with the series, we find that it
vanishes for $\lambda\to0$. Therefore the classical limit fixes $D_{0}=\lambda^{-1}(1-\mathcal{E})$.

We can repeat the same argument for $D_{1}$, for which we write
\[
(\rho(D_{1})\psi)(x)=\lambda^{-1}\left(c_{1}^{\prime}\psi(x)+c_{2}^{\prime}\psi(x_{0}-i\lambda,x_{1})\right)-ic_{3}^{\prime}(\partial_{1}\psi)(x)\ .
\]
This operator must reduce to $-i\partial_{1}$ in the limit $\lambda\to0$.
From the previous discussion it is obvious that we must have $c_{1}^{\prime}=c_{2}^{\prime}=0$.
Finally the third coefficient is the same as in the commutative case,
that is $c_{3}^{\prime}=1$, therefore $D_{1}=P_{1}$.
\end{proof}
The uniqueness of $\mathcal{D}$ and $\sigma$, under the symmetry
and classical limit assumptions, is the main result of this section.
We have a nice compatibility between the twisting of the commutator
and the modular properties of the algebra. Indeed the modular operator
$\Delta_{\omega}$, associated to the weight $\omega$, implements
the twist in the Hilbert space $\mathcal{H}$, that is $\pi(\sigma(f))=\Delta_{\omega}\pi(f)\Delta_{\omega}^{-1}$.
Moreover $\sigma$ is the analytic extension at $t=-i$ of $\sigma_{t}^{\omega}$,
the modular group of $\omega$, compare with the discussion on the
index map given in \cite{type III}.

Another interesting feature is the simple relation between $\mathcal{D}^{2}$,
the square of the Dirac operator, and $\mathcal{C}$, the \textit{first
Casimir} of the $\kappa$-Poincaré algebra (more precisely of the
quantum Euclidean group, since we are in Euclidean signature). The
latter is given by
\begin{equation}
\mathcal{C}=\frac{4}{\lambda^{2}}\sinh^{2}\left(\frac{\lambda\hat{P}_{0}}{2}\right)+e^{\lambda\hat{P}_{0}}\hat{P}_{1}^{2}\ .
\end{equation}
Then an elementary computation shows that we have the relation $\mathcal{D}^{2}=\Delta_{\omega}\mathcal{C}$,
where $\Delta_{\omega}$ is the modular operator of the weight $\omega$.
Apart from the presence of the modular operator, this is the same
property that one has in the commutative case. This connection is
made more suggestive if, following \cite{quantum-twisted}, we rewrite
the twisted commutator in the form
\[
K^{-1}\left(\mathcal{D}^{\prime}\pi(a)-(K^{-1}\pi(a)K)\mathcal{D}^{\prime}\right)\ .
\]
The usual definition of the twisted commutator is obtained by setting
$\mathcal{D}=K^{-1}\mathcal{D}^{\prime}$ and $\pi(\sigma(a))=K^{-2}\pi(a)K^{2}$.
From the previous remark it follows that $K=\Delta_{\omega}^{-1/2}$,
so we have $(\mathcal{D}^{\prime})^{2}=\mathcal{C}$ and $\mathcal{D}^{\prime}$
is exactly the {}``square root'' of the Casimir operator.

We remark that the operator $\mathcal{D}$, which we introduced in
this section, serves the purpose of describing the geometry of $\kappa$-Minkowski
space from the spectral point of view. This is in principle distinct
from the operator one introduces to describe the physical properties
of fermions, which is the one that deserves to be called the Dirac
operator. This has been studied in the literature, see for example
\cite{nowicki-dirac,dirac-spinors} and also \cite{dandrea-kappa}.
Mathematically such an operator is required to be equivariant (or
covariant, in more physical terms) under the $\kappa$-Poincaré algebra
(the quantum Euclidean group, in Euclidean signature). It is given
by
\[
\mathcal{D}^{eq}=\Gamma^{0}\left(\frac{1}{\lambda}\sinh(\lambda\hat{P}_{0})-\frac{\lambda}{2}e^{\lambda\hat{P}_{0}}\hat{P}_{1}^{2}\right)+\Gamma^{1}e^{\lambda\hat{P}_{0}}\hat{P}_{1}\ .
\]
In our notations it can be written as $\mathcal{D}^{eq}=\Gamma^{\mu}\rho(D_{\mu}^{eq})$,
where
\[
D_{0}^{eq}=\frac{1}{2\lambda}(\mathcal{E}^{-1}-\mathcal{E})-\frac{\lambda}{2}\mathcal{E}^{-1}P_{1}^{2}\ ,\quad D_{1}^{eq}=\mathcal{E}^{-1}P_{1}\ .
\]
By our previous results it follows that such an operator does not
have a twisted bounded commutator. A relevant algebraic property that
$\mathcal{D}^{eq}$ satisfies is $(\mathcal{D}^{eq})^{2}=\mathcal{C}+\frac{\lambda^{2}}{4}\mathcal{C}^{2}$,
where $\mathcal{C}$ is again the Casimir of the $\kappa$-Poincaré
algebra. Differently from the commutative case, we do not have that
the operator $\mathcal{D}^{eq}$ is the {}``square root'' of the
Casimir operator $\mathcal{C}$. On the other hand, as we remarked
above, this role is essentially played by the operator $\mathcal{D}$.

The most obvious difference between $\mathcal{D}$ and $\mathcal{D}^{eq}$
is that the former is equivariant only under the extended momentum
algebra $\mathcal{T}_{\kappa}$, while the latter is equivariant under
the full $\kappa$-Poincaré algebra $\mathcal{P}_{\kappa}$. However
we point out that only the subalgebra $\mathcal{T}_{\kappa}$ is relevant
for the introduction of $\kappa$-Minkowski space, so at least the
minimal requirement of equivariance under $\mathcal{T}_{\kappa}$
is satisfied. Moreover, for any $h\in\mathcal{P}_{\kappa}$, we have
the interesting property that the twisted commutator of $\mathcal{D}^{2}$
with $\rho(h)$ is zero. This follows from a one-line computation
\[
[\mathcal{D}^{2},\rho(h)]_{\sigma}=\mathcal{D}^{2}\rho(h)-\Delta_{\omega}\rho(h)\Delta_{\omega}^{-1}\mathcal{D}^{2}=\Delta_{\omega}(\mathcal{C}\rho(h)-\rho(h)\mathcal{C})=0\ .
\]

%% file: Trace.tex
In this section we consider the summability properties of our spectral
triple. We are going to show that it is not finitely summable in usual
sense of spectral triples, but it is finitely summable if we adapt
some definitions from the framework of \textit{modular spectral triples}.
This an extension of the concept of spectral triple, introduced among
other reasons to handle algebras having a KMS state. Since, as we
have seen in the previous sections, there is a natural KMS weight
on the algebra $\mathcal{A}$, it seems appropriate to use these tools
in this case.

The main result of this section is that the spectral dimension, computed
using the weight $\Phi$, exists and is equal to the classical dimension
two. Moreover the residue at $s=2$ of the function $\Phi\left(\pi(f)(\mathcal{D}^{2}+\mu^{2})^{-s/2}\right)$,
for $f\in\mathcal{A}$ and $\mu\neq0$, exists and gives $\omega(f)$
up to a constant, which shows that we recover the notion of integration
given by $\omega$ using the operator $\mathcal{D}$.

\subsection{A problem with finite summability}

The concept of finite summability for a non-unital spectral triple
is far more subtle than in the unital case, see \cite{integration}
for a detailed discussion of some of the issues arising. We just point
out that, while in the unital case the definition of the operator
$\mathcal{D}$ is enough to characterize the spectral dimension, in
the non-unital case one needs a delicate interplay between $\mathcal{D}$
and the algebra $\mathcal{A}$. We consider the notions of summability
given in \cite{integration}.
\begin{defn}
Let $(\mathcal{A},\mathcal{H},\mathcal{D})$ be a non-compact spectral
triple. We say that it is finitely summable and call $p$ the spectral
dimension if the following quantity exists 
\[
p:=\inf\{s>0:\forall a\in\mathcal{A},a\geq0,\ \text{Tr}\left(\pi(a)(\mathcal{D}^{2}+1)^{-s/2}\right)<\infty\}\ .
\]
In addition we say that $(\mathcal{A},\mathcal{H},\mathcal{D})$ is
$\mathcal{Z}_{p}$-summable if for all $a\in\mathcal{A}$ we have
\[
\limsup_{s\downarrow p}\left|(s-p)\text{Tr}\left(\pi(a)(\mathcal{D}^{2}+1)^{-s/2}\right)\right|<\infty\ .
\]
Now we show that our spectral triple is not finitely summable in this
sense.\end{defn}
\begin{prop}
\label{prop:not-compact}Let $h\in\mathcal{A}$ such that $h=f\star f$
with $f\geq0$. Then the operator $\pi(h)(\mathcal{D}^{2}+1)^{-s/2}$
is not trace class for any $s>0$. In other words, the spectral triple
is not finitely summable.\end{prop}
\begin{proof}
We have that if $\pi(h)(\mathcal{D}^{2}+1)^{-s/2}$ is trace class
then also $\pi(f)(\mathcal{D}^{2}+1)^{-s/2}\pi(f)$ is trace class,
while the converse statement is not true in general, see the discussion
in \cite{integration}. Proving that $\pi(f)(\mathcal{D}^{2}+1)^{-s/2}\pi(f)$
is trace class is the same as proving that $\pi(f)(\mathcal{D}^{2}+1)^{-s/4}$
is Hilbert-Schmidt, which is easy to check using the integral formula
for the kernel. Now we show that $\pi(f)(\mathcal{D}^{2}+1)^{-s/4}$
is not Hilbert-Schmidt for any $s>0$, from which the proposition
follows.

We have that the Hilbert-Schmidt norm of $\pi(f)(\mathcal{D}^{2}+1)^{-s/4}$,
as an operator on $\mathcal{H}=\mathcal{H}_{r}\otimes\mathbb{C}^{2}$,
is equal to the Hilbert-Schmidt norm of $A:=U\pi(f)(\mathcal{D}^{2}+1)^{-s/4}U^{-1}$
as an operator on $L^{2}(\mathbb{R}^{2})\otimes\mathbb{C}^{2}$. Using
Proposition \ref{prop:kernel} we find that the Schwartz kernel of
$A$ is given by
\[
K_{A}(x,y)=\int e^{ip(x-y)}(Uf)(x_{0},e^{\lambda p_{0}}x_{1})G_{s}(p_{0},e^{-\lambda p_{0}}p_{1})\frac{d^{2}p}{(2\pi)^{2}}\ ,
\]
where the function $G_{s}$ is defined by
\[
G_{s}(p):=\left(\lambda^{-2}(1-e^{-\lambda p_{0}})^{2}+p_{1}^{2}+1\right)^{-s/4}\ .
\]
For fixed $x$ define the function $h_{x}(p):=(Uf)(x_{0},e^{\lambda p_{0}}x_{1})G_{s}(p_{0},e^{-\lambda p_{0}}p_{1})$.
With this definition we can write the kernel $K_{A}$ as an inverse
Fourier transform
\[
K_{A}(x,y)=\int e^{ip(x-y)}h_{x}(p)\frac{d^{2}p}{(2\pi)^{2}}=(\mathcal{F}^{-1}h_{x})(x-y)\ .
\]
Now it is easy to compute the Hilbert-Schmidt norm of $A$. We have
\[
\begin{split}\|A\|_{2}^{2} & =2\int\int|K_{A}(x,y)|^{2}d^{2}xd^{2}y=2\int\int|(\mathcal{F}^{-1}h_{x})(x-y)|^{2}d^{2}xd^{2}y\\
 & =2\int\int|(\mathcal{F}^{-1}h_{x})(y)|^{2}d^{2}xd^{2}y\ .
\end{split}
\]
The factor $2$ comes from the dimension of the spinor bundle, since
$\mathcal{H}=\mathcal{H}_{r}\otimes\mathbb{C}^{2}$. Now, using the
fact that the Fourier transform is a unitary operator in $L^{2}(\mathbb{R}^{2})$
(up to the factor $(2\pi)^{2}$, in our conventions), we can rewrite
the previous expression as
\[
\|A\|_{2}^{2}=2\int\int|h_{x}(p)|^{2}d^{2}x\frac{d^{2}p}{(2\pi)^{2}}=2\int\int|(Uf)(x_{0},e^{\lambda p_{0}}x_{1})G_{s}(p_{0},e^{-\lambda p_{0}}p_{1})|^{2}d^{2}x\frac{d^{2}p}{(2\pi)^{2}}\ .
\]
After the change of variables $x_{1}\to e^{-\lambda p_{0}}x_{1}$,
$p_{1}\to e^{\lambda p_{0}}p_{1}$ we find
\begin{equation}
\|A\|_{2}^{2}=2\int|(Uf)(x)|^{2}d^{2}x\int|g_{s}(p)|^{2}\frac{d^{2}p}{(2\pi)^{2}}=\frac{2}{(2\pi)^{2}}\|Uf\|_{2}^{2}\|G_{s}\|_{2}^{2}\ .\label{eq:hs-norm}
\end{equation}
Now consider the norm $\|G_{s}\|_{2}$, which is given by the expression
\[
\|G_{s}\|_{2}^{2}=\int\left(\lambda^{-2}(1-e^{-\lambda p_{0}})^{2}+p_{1}^{2}+1\right)^{-s/2}d^{2}p\ .
\]
Notice that the integrand does not go to zero for $p_{0}\to\infty$,
so we find that $\|G_{s}\|_{2}=\infty$ for any $s>0$. Therefore
the operator $\pi(f)(1+\mathcal{D}^{2})^{-s/4}$ is not Hilbert-Schmidt
for any $s>0$.
\end{proof}
Now we argue that using the framework of twisted spectral triples
is not enough to describe the non-commutative geometry of $\kappa$-Minkowski
space, and that some more refined notion is needed to capture the
modular properties associated to this geometry. First we recall a
result obtained in \cite{type III}: consider a twisted spectral triple
($\mathcal{A},\mathcal{H},\mathcal{D})$ with twist $\sigma$ and
such that $\mathcal{D}^{-1}\in\mathcal{L}^{n+}$. Define the linear
functional $\varphi(a)=\text{Tr}_{\omega}(\pi(a)\mathcal{D}^{-n})$,
where $\text{Tr}_{\omega}$ is the Dixmier trace. Then for any $a,b\in\mathcal{A}$
we have $\varphi(ab)=\varphi(\sigma^{n}(b)a)$. Putting aside the
issues of the non-compact case, which are not needed for this heuristic
argument, having a spectral dimension $n=2$ in our case would imply
$\varphi(f\star g)=\varphi(\sigma^{2}(g)\star f).$ The KMS condition
for $\omega$, on the other hand, can be rewritten in terms of the
twist $\sigma$ and reads $\omega(f\star g)=\omega(\sigma(g)\star f)$.
Since we expect the Dixmier trace to be linked to the weight $\omega$,
we see that there is a tension between the two notions of integration,
which are due to their different modular properties.

\subsection{Modular spectral triples}

Similar problems appear quite generically for algebras involving KMS
states. To deal with them the framework of \textit{modular spectral
triples} was introduced, which has emerged from a series of papers
\cite{modular1,modular2,modular3,modular4}, see also the review \cite{modular-review}.
A modification of this idea involving twisted commutators has been
also considered in \cite{suq2-kaad,kaad}. The main point of a modular
version of a spectral triple is to enable the use of a weight, instead
of a trace, to measure the growth of the resolvent of the operator
$\mathcal{D}$. The definition of a modular spectral triple is usually
given using the language of semi-finite spectral triples, a generalization
of spectral triples introduced to deal with the case of a semi-finite
von Neumann algebra $\mathcal{N}$. We do not need this extra technology,
as we are going to be concerned with the case $\mathcal{N}=\mathcal{B}(\mathcal{H})$.

However, as we remarked in the introduction, our model does not fit
the definition of modular spectral triple, as given for example in
\cite{modular4}. Indeed this notion has been formalized on the base
of examples where the modular group associated to the KMS state is
periodic and, in this situation, it makes sense to restrict some conditions
to the fixed point algebra under this action. This possibility is
not available in our case: as we have seen in the previous sections,
the weight $\omega$ is a KMS weight with respect to the action $(\sigma_{t}^{\omega}f)(x)=f(x_{0}-\lambda t,x_{1})$,
that is translation in the first variable. The fixed points for this
action are functions that are constant in the first variable and,
since the functions in $\mathcal{A}$ vanish at infinity, the only
fixed point is given by the zero function.

This fact implies that, strictly speaking, our construction does not
fit into this framework. Here we are not going to be concerned whether
it might be extended to cover this case. More humbly, we are going
to borrow some ingredients from this framework, and show that we obtain
sensible results by applying them to our case. More specifically we
use the notion of spectral dimension involving a weight $\Phi$, which
reduces to the usual one in the case when $\Phi$ is the operator
trace. The definition we use is the following.
\begin{defn}
\label{def:sum-phi}Let $(\mathcal{A},\mathcal{H},\mathcal{D})$ be
a non-compact modular spectral triple with weight $\Phi$. We say
that it is finitely summable and call $p$ the spectral dimension
if the following quantity exists 
\[
p:=\inf\{s>0:\forall a\in\mathcal{A},a\geq0,\ \Phi\left(\pi(a)(\mathcal{D}^{2}+1)^{-s/2}\right)<\infty\}\ .
\]

\end{defn}
We do not introduce the notion of $\mathcal{Z}_{p}$-summability,
since the ideals $\mathcal{Z}_{p}$ are defined for traces and not
for weights \cite{integration}.

We now argue why such a notion is useful in our case. Denote by $\varphi$
the non-commutative integral defined in terms of the Dixmier trace.
Since we have shown that the spectral triple is not finitely summable,
it follows that the non-commutative integral does not exist. However,
as we remarked at the end of the previous subsection, if the spectral
dimension were equal to two it would follow, from the general result
for twisted spectral triples, that $\varphi(f\star g)=\varphi(\sigma^{2}(g)\star f)$.
On the other hand, the KMS condition for the weight $\omega$ can
be written in terms of the twist $\sigma$ and reads $\omega(f\star g)=\omega(\sigma(g)\star f)$.
Since the non-commutative integral should recover the notion of integration
on the algebra, given by the weight $\omega$, we see that we have
a mismatch of modular properties. We point out that the twist $\sigma$
is implemented by $\Delta_{\omega}$, the modular operator of $\omega$,
and this fact hints to the possibility of correcting the non-commutative
integral by inserting $\Delta_{\omega}$ appropriately in the definition.
To this end we consider the weight $\Phi$ defined by $\Phi(\cdot):=\mathrm{Tr}(\Delta_{\omega}\cdot)$,
and use it to compute the spectral dimension.

The main result of this section is that the spectral dimension, computed
using the weight $\Phi$, exists and is equal to the classical dimension
two. Moreover the residue at $s=2$ of the function $\Phi\left(\pi(f)(\mathcal{D}^{2}+\mu^{2})^{-s/2}\right)$,
for $f\in\mathcal{A}$ and $\mu\neq0$, exists and gives $\omega(f)$
up to a constant, which shows that we recover the notion of integration
given by $\omega$ using the operator $\mathcal{D}$. This gives an
analogue of the $\mathcal{Z}_{2}$-summability condition, and is in
line with similar results obtained for the modular spectral triples
studied so far.

Before starting the computation we note the following easy but useful
lemma.
\begin{lem}
\label{lem:form-weight}For all $f\in\mathcal{A}$ we have $\Phi\left(\pi(f)(\mathcal{D}^{2}+1)^{-s/2}\right)=\mathrm{Tr}\left(\pi(\sigma_{-i}^{\omega}f)\Delta_{\omega}(\mathcal{D}^{2}+1)^{-s/2}\right)$.\end{lem}
\begin{proof}
It is easily proven by the following computation
\[
\begin{split}\Phi\left(\pi(f)(\mathcal{D}^{2}+1)^{-s/2}\right) & =\mathrm{Tr}\left(\Delta_{\omega}\pi(f)\Delta_{\omega}^{-1}\Delta_{\omega}(\mathcal{D}^{2}+1)^{-s/2}\right)\\
 & =\mathrm{Tr}\left(\pi(\sigma_{-i}^{\omega}f)\Delta_{\omega}(\mathcal{D}^{2}+1)^{-s/2}\right)\ .
\end{split}
\]
In the last line we have used the fact that $\sigma^{\omega}$ is
implemented by the modular operator $\Delta_{\omega}$, that is $\pi(\sigma_{t}^{\omega}f)=\Delta_{\omega}^{it}\pi(f)\Delta_{\omega}^{-it}$
for any $f\in\mathcal{A}$.
\end{proof}
The next subsection is devoted to proving the results announced above.

\subsection{The spectral dimension with the weight $\Phi$}

We can restrict our attention to the operator $\pi(f)\Delta_{\omega}(\mathcal{D}^{2}+1)^{-s/2}$
on $\mathcal{H}$ and, via the unitary operator $U$, to the operator
$U\pi(f)\Delta_{\omega}(\mathcal{D}^{2}+1)^{-s/2}U^{-1}$ on $L^{2}(\mathbb{R}^{2})\otimes\mathbb{C}^{2}$.
Now we want to see if this operator is trace class for some $s>0$
and compute its trace. To prove that an operator $A$ on $L^{2}(\mathbb{R}^{n})$
is trace class, a possible strategy is to show that it is a pseudo-differential
operator of sufficiently negative order (see for example \cite[Chapter IV]{shubin}).
We say that $A$ is a pseudo-differential operator of order $m$ if
its symbol $a$ satisfies the condition $|\partial_{x}^{\alpha}\partial_{\xi}^{\beta}a(x,\xi)|\leq c_{\alpha\beta}(1+|\xi|)^{m-|\alpha|}$,
where $c_{\alpha\beta}$ are constants and we use the multi-index
notation. However this class of symbols is not well adapted to the
present situation, as we now argue.

Using Proposition \ref{prop:kernel} we know that the symbol of $U\pi(f)\Delta_{\omega}(\mathcal{D}^{2}+1)^{-s/2}U^{-1}$
is of the form $(Uf)(x_{0},e^{\lambda\xi_{0}}x_{1})g(\xi)$, for some
function $g$. Now consider the derivative with respect to $x_{1}$,
which is given by $e^{\lambda\xi_{0}}(\partial_{1}Uf)(x_{0},e^{\lambda\xi_{0}}x_{1})g(\xi)$.
By examining the behaviour at $\xi_{0}\to\infty$ we see that we can
not bound this function uniformly in $x_{1}$. Consider first the
case $x_{1}\neq0$: by defining $y_{1}=e^{\lambda\xi_{0}}x_{1}$,
we have $x_{1}^{-1}y_{1}(\partial_{1}Uf)(x_{0},y_{1})$ and this goes
to zero for $y_{1}\to\pm\infty$, since $Uf$ is a Schwartz function.
For $x_{1}=0$, however, we get $e^{\lambda\xi_{0}}(\partial_{1}Uf)(x_{0},0)$,
which grows exponentially in $\xi_{0}$. This implies that at $x_{1}=0$
we can not satisfy the condition for a pseudo-differential operator
of negative order, and so we can not use the related results.

Since we have this problem only at $x_{1}=0$, which is a set of measure
zero in $\mathbb{R}^{2}$, we can expect to be able to overcome this
problem by using a criterion which involves an $L^{1}$ condition
on the symbol. To this end we are going to use the following theorem
given in \cite{arsu}. Let $A$ be an operator in $L^{2}(\mathbb{R}^{n})$
with symbol $a(x,\xi)$. If the symbol satisfies the condition $\partial_{x}^{\alpha}\partial_{\xi}^{\beta}a\in L^{p}(\mathbb{R}^{n}\times\mathbb{R}^{n})$
for $|\alpha|,|\beta|\leq[n/2]+1$, where $1\leq p<\infty$, then
$A$ belongs to the $p$-th Schatten ideal in $L^{2}(\mathbb{R}^{n})$.
Here we are using the multi-index notation for $\alpha,\beta$ and
$[n]$ denotes the integer part of $n$. Using this result we can
prove the following.
\begin{thm}
Let $f\in\mathcal{A}$ and $\mu\neq0$. Then the operator $\pi(f)\Delta_{\omega}(\mathcal{D}^{2}+\mu^{2})^{-s/2}$
is trace class for $s>2$. In particular we have spectral dimension
$p=2$ according to Definition \ref{def:sum-phi}.\end{thm}
\begin{proof}
Using Proposition \ref{prop:kernel} we have that the symbol of the
operator $U\pi(f)\Delta_{\omega}(\mathcal{D}^{2}+\mu^{2})^{-s/2}U^{-1}$
is given by $a(x,\xi):=(Uf)(x_{0},e^{\lambda\xi_{0}}x_{1})G_{s}^{\Delta}(\xi_{0},e^{-\lambda\xi_{0}}\xi_{1})$,
where
\[
G_{s}^{\Delta}(\xi):=e^{-\lambda\xi_{0}}\left(\lambda^{-2}(1-e^{-\lambda\xi_{0}})^{2}+\xi_{1}^{2}+\mu^{2}\right)^{-s/2}\ .
\]
As we remarked above, to prove that this operator is trace class we
are going to show that the symbol satisfies the condition $\partial_{x}^{\alpha}\partial_{\xi}^{\beta}a\in L^{1}(\mathbb{R}^{2}\times\mathbb{R}^{2})$,
for $|\alpha|,|\beta|\leq2$. Let us start by showing that $a$ is
integrable. Using a change of variables we get
\[
\begin{split}\int\int|a(x,\xi)|d^{2}xd^{2}\xi & =\int\int|(Uf)(x_{0},e^{\lambda\xi_{0}}x_{1})|G_{s}^{\Delta}(\xi_{0},e^{-\lambda\xi_{0}}\xi_{1})d^{2}xd^{2}\xi\\
 & =\int|(Uf)(x)|d^{2}x\int G_{s}^{\Delta}(\xi)d^{2}\xi=\|Uf\|_{1}\|G_{s}^{\Delta}\|_{1}\ .
\end{split}
\]
We have that $\|Uf\|_{1}$ is finite since $Uf$ is a Schwartz function.
To prove that $\|G_{s}^{\Delta}\|_{1}$ is finite we need to consider
the asymptotic behaviour of the function $G_{s}^{\Delta}(\xi)$. This
is given by
\[
\begin{array}{cc}
G_{s}^{\Delta}(\xi)\sim e^{-\lambda\xi_{0}}|\xi_{1}|^{-s}\ , & \xi_{0}\to\infty\ ,\ |\xi_{1}|\to\infty\ ,\\
G_{s}^{\Delta}(\xi)\sim e^{-\lambda\xi_{0}}\left(e^{-2\lambda\xi_{0}}+\xi_{1}^{2}\right)^{-s/2}\ , & \xi_{0}\to-\infty\ ,\ |\xi_{1}|\to\infty\ .
\end{array}
\]
For $\xi_{0}\to\infty$ this function is integrable if $s>1$. To
study the other case we use the integral
\[
\int(c^{2}+x^{2})^{-s/2}dx=\sqrt{\pi}\frac{\Gamma\left(\frac{s-1}{2}\right)}{\Gamma\left(\frac{s}{2}\right)}(c^{2})^{-\frac{s-1}{2}}\ .
\]
Then the function $G_{s}^{\Delta}$ is integrable for $\xi_{0}\to-\infty$
is $s>2$, as we have
\[
\int G_{s}^{\Delta}(\xi)d\xi_{1}\sim e^{-\lambda\xi_{0}}\int\left(e^{-2\lambda\xi_{0}}+\xi_{1}^{2}\right)^{-s/2}d\xi_{1}\sim e^{-\lambda\xi_{0}}e^{(s-1)\lambda\xi_{0}}\ .
\]

Now we consider the derivatives in $x$ and show that the integral
of $\partial_{x}^{\alpha}a$ vanishes, since $Uf$ is a Schwartz function.
In the following we use the notation $\partial_{0}$ and $\partial_{1}$
to denote derivatives with respect to the first and second variable.
For the derivatives in $x_{0}$ we have 
\[
\int\partial_{x_{0}}^{n}(Uf)(x_{0},e^{\lambda\xi_{0}}x_{1})dx_{0}=(\partial_{0}^{n-1}Uf)(\infty,e^{\lambda\xi_{0}}x_{1})-(\partial_{0}^{n-1}Uf)(-\infty,e^{\lambda\xi_{0}}x_{1})=0\ .
\]
Similarly for the derivatives in $x_{1}$ we obtain
\[
\begin{split}\int\partial_{x_{1}}^{n}(Uf)(x_{0},e^{\lambda\xi_{0}}x_{1})dx_{1} & =e^{n\lambda\xi_{0}}\int(\partial_{1}^{n}Uf)(x_{0},e^{\lambda\xi_{0}}x_{1})dx_{1}=e^{(n-1)\lambda\xi_{0}}\int(\partial_{1}^{n}Uf)(x)dx_{1}\\
 & =e^{(n-1)\lambda\xi_{0}}\left((\partial_{1}^{n-1}Uf)(x_{0},\infty)-(\partial_{1}^{n-1}Uf)(x_{0},-\infty)\right)=0\ .
\end{split}
\]

Now we consider the derivatives with respect to $\xi_{0}$. There
are two contributions, one coming from $(Uf)(x_{0},e^{\lambda\xi_{0}}x_{1})$
and the other from $G_{s}^{\Delta}(\xi_{0},e^{-\lambda\xi_{0}}\xi_{1})$.
First we will consider derivatives acting on the term $(Uf)(x_{0},e^{\lambda\xi_{0}}x_{1})$.
Taking one derivative we obtain
\begin{equation}
\begin{split}\partial_{\xi_{0}}(Uf)(x_{0},e^{\lambda\xi_{0}}x_{1}) & =e^{\lambda\xi_{0}}\lambda x_{1}(\partial_{1}Uf)(x_{0},e^{\lambda\xi_{0}}x_{1})\\
 & =\lambda x_{1}\partial_{x_{1}}(Uf)(x_{0},e^{\lambda\xi_{0}}x_{1})\\
 & =\lambda\partial_{x_{1}}\left(x_{1}(Uf)(x_{0},e^{\lambda\xi_{0}}x_{1})\right)-\lambda(Uf)(x_{0},e^{\lambda\xi_{0}}x_{1})\ .
\end{split}
\label{eq:derivative-xi}
\end{equation}
The second term corresponds to the case with zero derivatives, so
we have already proven that it gives a finite contribution. The first
term on the other hand vanishes upon integration in $x_{1}$, since
$Uf$ is a Schwartz function. For the second derivative using (\ref{eq:derivative-xi})
we obtain
\[
\begin{split}\partial_{\xi_{0}}^{2}(Uf)(x_{0},e^{\lambda\xi_{0}}x_{1}) & =\lambda x_{1}\partial_{x_{1}}\partial_{\xi_{0}}(Uf)(x_{0},e^{\lambda\xi_{0}}x_{1})\\
 & =\lambda\partial_{x_{1}}\left(x_{1}\partial_{\xi_{0}}(Uf)(x_{0},e^{\lambda\xi_{0}}x_{1})\right)-\lambda\partial_{\xi_{0}}(Uf)(x_{0},e^{\lambda\xi_{0}}x_{1})
\end{split}
\]
The second term corresponds to the case of one derivative, which we
have proven to be finite. For the first term we just need to notice
that, for fixed $\xi_{0}$, equation (\ref{eq:derivative-xi}) implies
that $\partial_{\xi_{0}}(Uf)(x_{0},e^{\lambda\xi_{0}}x_{1})$ is a
Schwartz function. Then the first term vanishes by the previous argument.
\end{proof}
Since we have established that the operator $\pi(f)\Delta_{\omega}(\mathcal{D}^{2}+\mu^{2})^{-s/2}$
is trace class we can now compute its trace. We are going to show
that the residue at $s=2$ recovers, up to a constant, the weight
$\omega$ on the algebra $\mathcal{A}$.
\begin{prop}
Let $f\in\mathcal{A}$ and $\mu\neq0$. Then we have\textup{
\[
\lim_{s\to2}(s-2)\Phi\left(\pi(f)(\mathcal{D}^{2}+\mu^{2})^{-s/2}\right)=\frac{1}{2\pi}\omega(f)\ .
\]
}\end{prop}
\begin{proof}
Using Lemma \ref{lem:form-weight} it suffices to prove the analogue
result with $\pi(f)\Delta_{\omega}(1+\mathcal{D}^{2})^{-s/2}$. In
the previous theorem we have shown that the operator $A:=U\pi(f)\Delta_{\omega}(\mathcal{D}^{2}+\mu^{2})^{-s/2}U^{-1}$,
with $\mu\neq0$, is trace class for $s>2$. We can compute its trace
by integrating the kernel, that is
\[
\text{Tr}\left(\pi(f)\Delta_{\omega}(\mathcal{D}^{2}+\mu^{2})^{-s/2}\right)=2\int\int(Uf)(x_{0},e^{\lambda\xi_{0}}x_{1})G_{s}^{\Delta}(\xi_{0},e^{-\lambda\xi_{0}}\xi_{1})d^{2}x\frac{d^{2}\xi}{(2\pi)^{2}}\ .
\]
Here the factor $2$ comes from the dimension of the spinor bundle,
since $\mathcal{H}=\mathcal{H}_{r}\otimes\mathbb{C}^{2}$. As shown
in the previous theorem this is actually equal to
\[
\text{Tr}\left(\pi(f)\Delta_{\omega}(\mathcal{D}^{2}+\mu^{2})^{-s/2}\right)=2\int(Uf)(x)d^{2}x\int G_{s}^{\Delta}(\xi)\frac{d^{2}\xi}{(2\pi)^{2}}\ .
\]

Now we need to compute the integral in $\xi$, which is given by
\[
c(s):=\int G_{s}^{\Delta}(\xi)\frac{d^{2}\xi}{(2\pi)^{2}}=\int e^{-\lambda\xi_{0}}\left(\lambda^{-2}\left(1-e^{-\lambda\xi_{0}}\right)^{2}+\xi_{1}^{2}+\mu^{2}\right)^{-s/2}\frac{d^{2}\xi}{(2\pi)^{2}}\ .
\]
The integral over $\xi_{1}$ can be easily computed using the standard
formula 
\[
\int(x^{2}+a^{2})^{-s/2}dx=\sqrt{\pi}\frac{\Gamma\left(\frac{s-1}{2}\right)}{\Gamma\left(\frac{s}{2}\right)}(a^{2})^{-\frac{s-1}{2}}\ ,\qquad s>1\ .
\]
Using this result we have
\[
c(s)=\frac{\sqrt{\pi}}{(2\pi)^{2}}\frac{\Gamma\left(\frac{s-1}{2}\right)}{\Gamma\left(\frac{s}{2}\right)}\int e^{-\lambda\xi_{0}}\left(\lambda^{-2}\left(1-e^{-\lambda\xi_{0}}\right)^{2}+\mu^{2}\right)^{-\frac{s-1}{2}}d\xi_{0}\ .
\]
Now, using the change of variable $r=e^{-\lambda\xi_{0}}$, we rewrite
the integral in $\xi_{0}$ as
\[
\int e^{-\lambda\xi_{0}}\left(\lambda^{-2}\left(1-e^{-\lambda\xi_{0}}\right)^{2}+\mu^{2}\right)^{-\frac{s-1}{2}}d\xi_{0}=\lambda^{s-2}\int\left((1-r)^{2}+\lambda^{2}\mu^{2}\right)^{-\frac{s-1}{2}}dr\ .
\]
This integral can be solved analytically for $s>2$
\[
\int_{0}^{\infty}\left((1-r)^{2}+a^{2}\right)^{-\frac{s-1}{2}}dr=a^{1-s}\left(a\frac{\sqrt{\pi}}{2}\frac{\Gamma\left(\frac{s}{2}-1\right)}{\Gamma\left(\frac{s-1}{2}\right)}+{}_{2}F_{1}\left(\frac{1}{2},\frac{s-1}{2},\frac{3}{2},-\frac{1}{a^{2}}\right)\right)\ ,
\]
where $_{2}F_{1}(a,b,c,z)$ is the ordinary hypergeometric function.
The result is then
\[
c(s)=\frac{\sqrt{\pi}}{(2\pi)^{2}}\frac{\Gamma\left(\frac{s-1}{2}\right)}{\Gamma\left(\frac{s}{2}\right)}\lambda^{s-2}(\lambda\mu)^{1-s}\left(\lambda\mu\frac{\sqrt{\pi}}{2}\frac{\Gamma\left(\frac{s}{2}-1\right)}{\Gamma\left(\frac{s-1}{2}\right)}+{}_{2}F_{1}\left(\frac{1}{2},\frac{s-1}{2},\frac{3}{2},-\frac{1}{\lambda^{2}\mu^{2}}\right)\right)\ .
\]

Now we can easily show the analogue of $\mathcal{Z}_{2}$-summability,
that is
\[
\limsup_{s\downarrow2}\left|(s-2)\text{Tr}\left(\pi(f)\Delta_{\omega}(\mathcal{D}^{2}+\mu^{2})^{-s/2}\right)\right|<\infty\ .
\]
We need to study the behaviour of the function $c(s)$ around $s=2$.
Notice that the second term, the one involving the hypergeometric
function, is regular at $s=2$, while the first term has a simple
pole at $s=2$, which comes from the function $\Gamma\left(\frac{s}{2}-1\right)$.
Indeed the Laurent expansion of this function at $s=2$ is given by
\[
\Gamma\left(\frac{s}{2}-1\right)=\frac{2}{s-2}-\gamma+O(s-2)\ ,
\]
where $\gamma$ is the Euler-Mascheroni constant. Using this fact
we have
\[
\begin{split}\limsup_{s\downarrow2}(s-2)c(s) & =\limsup_{s\downarrow2}(s-2)\frac{\sqrt{\pi}}{(2\pi)^{2}}\frac{1}{\Gamma\left(\frac{s}{2}\right)}\lambda^{s-2}(\lambda\mu)^{1-s}\lambda\mu\frac{\sqrt{\pi}}{2}\Gamma\left(\frac{s}{2}-1\right)\\
 & =\frac{\sqrt{\pi}}{(2\pi)^{2}}(\lambda\mu)^{-1}\lambda\mu\frac{\sqrt{\pi}}{2}2=\frac{1}{4\pi}\ .
\end{split}
\]
Notice that the result of this limit does not depend on $\mu$, as
expected. Then we have 
\[
\lim_{s\to2}(s-2)\Phi\left(\pi(f)\Delta_{\omega}(1+\mathcal{D}^{2})^{-s/2}\right)=\frac{1}{2\pi}\int\sigma_{-i}^{\omega}(f)(x)d^{2}x\ .
\]
Using the properties of the functions in $\mathcal{A}$ and the Cauchy
theorem we have that
\[
\int f(x_{0}+z,x_{1})d^{2}x=\int f(x)d^{2}x\ ,
\]
for any $z\in\mathbb{C}$. Then the result follows, since $\sigma_{-i}^{\omega}(f)(x)=f(x_{0}+i\lambda,x_{1})$.\end{proof}

%% file: Real.tex
In this section we discuss the possibility of introducing a real structure
on our spectral triple. The outcome is that most conditions are trivially
satisfied, the first order condition is modified in a way that involves
the twisted commutator as in \cite{quantum-twisted}, and one condition
involving the Dirac operator is modified by the presence of the modular
operator $\Delta_{\omega}$.

We briefly review the commutative case, to set the notation and also
since some computations are going to be identical in the non-commutative
case. Consider a spectral triple $(\mathcal{A},\mathcal{H},\mathcal{D}$)
of spectral dimension $n$. It is even if there exists a $\mathbb{Z}_{2}$-grading
$\chi$ on $\mathcal{H}$.
\begin{defn}
A real structure for the spectral triple $(\mathcal{A},\mathcal{H},\mathcal{D}$)
is an antilinear isometry $\mathcal{J}:\mathcal{H}\to\mathcal{H}$
with the following properties:
\begin{enumerate}
\item $\mathcal{J}^{2}=\varepsilon(n)$,
\item $\mathcal{J}\mathcal{D}=\varepsilon^{\prime}(n)\mathcal{D}\mathcal{J}$,
\item $[\pi(f),\mathcal{J}\pi(g^{*})\mathcal{J}^{-1}]=0$,
\item $\mathcal{J}\chi=i^{n}\chi\mathcal{J}$ if it is even,
\item $[[\mathcal{D},\pi(f)],\mathcal{J}\pi(g)\mathcal{J}^{-1}]=0$.
\end{enumerate}
\end{defn}
The fifth condition is usually called the first order condition. Here
$\varepsilon(n)$ and $\varepsilon^{\prime}(n)$ are mod $8$ periodic
functions which are given by
\[
\begin{split}\varepsilon(n) & =(1,1,-1,-1,-1,-1,1,1)\ ,\\
\varepsilon^{\prime}(n) & =(1,-1,1,1,1,-1,1,1)\ .
\end{split}
\]

We shortly review the case $n=2$. Consider $\mathcal{J}=CJ_{c}$,
where $J_{c}$ is complex conjugation and $C$ is a $2\times2$ matrix,
which in our conventions is given by $C=i\Gamma^{0}$. The grading
is given by the matrix $\chi=-i\Gamma^{0}\Gamma^{1}$, which satisfies
$\chi^{2}=1$. The Hilbert space is $\mathcal{H}=L^{2}(\mathbb{R}^{2})\otimes\mathbb{C}^{2}$
and acting with the operator $\mathcal{J}$ on a spinor $\psi$ we
get
\[
\mathcal{J}\left(\begin{array}{c}
\psi_{1}(x)\\
\psi_{2}(x)
\end{array}\right)=\left(\begin{array}{c}
-i\overline{\psi_{2}}(x)\\
-i\overline{\psi_{1}}(x)
\end{array}\right)\ .
\]
The fact that that $\mathcal{J}$ is an antilinear isometry follows
from a simple computation
\[
\begin{split}(\mathcal{J}\phi,\mathcal{J}\psi)_{\mathcal{H}} & =\sum_{k=1}^{2}((\mathcal{J}\phi)_{k},(\mathcal{J}\psi)_{k})_{L^{2}(\mathbb{R}^{2})}=\sum_{k=1}^{2}(-i\overline{\phi_{k}},-i\overline{\psi_{k}})_{L^{2}(\mathbb{R}^{2})}\\
 & =\sum_{k=1}^{2}(\psi_{k},\phi_{k})_{L^{2}(\mathbb{R}^{2})}=(\psi,\phi)_{\mathcal{H}}\ .
\end{split}
\]
The first condition is easily verified using the properties of the
$\Gamma$ matrices. To verify the second condition write the Dirac
operator as $\mathcal{D}=\Gamma^{\mu}\hat{P}_{\mu}$. We have
\[
\mathcal{J}\mathcal{D}\mathcal{J}=i\Gamma^{0}J_{c}\Gamma^{\mu}\hat{P}_{\mu}i\Gamma^{0}J_{c}=\Gamma^{\mu}J_{c}\hat{P}_{\mu}J_{c}\ .
\]
Using the definition of $\hat{P}_{\mu}$ we get $\Gamma^{\mu}J_{c}\hat{P}_{\mu}J_{c}=-\Gamma^{\mu}\hat{P}_{\mu}=-\mathcal{D}$.
Applying $\mathcal{J}$ and using $\mathcal{J}^{2}=1$ we obtain $\mathcal{J}\mathcal{D}=-\mathcal{D}\mathcal{J}$.
For the third condition we note that $\mathcal{J}\pi(g^{*})\mathcal{J}^{-1}=i\Gamma^{0}J_{c}\overline{g}i\Gamma^{0}J_{c}=g$.
The fourth condition follows from the computation 
\[
\mathcal{J}\chi=i\Gamma^{0}J_{c}(-i)\Gamma^{0}\Gamma^{1}=i\Gamma^{0}(-i)\Gamma^{0}\Gamma^{1}J_{c}=-(-i\Gamma^{0}\Gamma^{1})(i\Gamma^{0}J_{c})=-\chi\mathcal{J}\ .
\]
Finally the fifth condition follows since $[\mathcal{D},\pi(f)]=\Gamma^{\mu}\pi(D_{\mu}f)$
and therefore
\[
[[\mathcal{D},\pi(f)],\mathcal{J}\pi(g)\mathcal{J}^{-1}]=\Gamma^{\mu}[\pi(D_{\mu}f),\mathcal{J}\pi(g)\mathcal{J}^{-1}]=0\ .
\]

Now we consider the non-commutative case. For $f\in\mathcal{A}$ define
the operator $\tilde{J}f:=\sigma_{i/2}^{\omega}(f^{*})$, see \cite{kms-weights}.
Since $\omega$ satisfies the KMS condition with respect to $\sigma^{\omega}$,
the term $\sigma_{i/2}^{\omega}$ compensates for the lack of the
trace property, as shown in the next lemma.
\begin{lem}
The operator $\tilde{J}$ is an antilinear isometry on $\mathcal{H}_{r}$.\end{lem}
\begin{proof}
Recall that $\sigma_{t}^{\omega}(f)(x)=f(x_{0}-\lambda t,x_{1})$
and that for $f,g\in\mathcal{A}$ we have the KMS property $\omega(f\star g)=\omega(\sigma_{-i}^{\omega}(g)\star f)$.
Then we have the following
\[
\begin{split}(\tilde{J}f,\tilde{J}g) & =\omega((\sigma_{i/2}^{\omega}(f^{*}))^{*}\star\sigma_{i/2}^{\omega}(g^{*}))=\omega(\sigma_{-i/2}^{\omega}(f)\star\sigma_{i/2}^{\omega}(g^{*}))\\
 & =\omega(\sigma_{-i}^{\omega}\sigma_{i/2}^{\omega}(g^{*})\star\sigma_{-i/2}^{\omega}(f))=\omega(\sigma_{-i/2}^{\omega}(g^{*})\star\sigma_{-i/2}^{\omega}(f))\\
 & =\omega(\sigma_{-i/2}^{\omega}(g^{*}\star f))=\omega(g^{*}\star f)=(g,f)\ .
\end{split}
\]
The property $\omega(\sigma_{-i/2}^{\omega}(f))=\omega(f)$ holds
for any $f\in\mathcal{A}$, as previously discussed. In particular
we have that $\|\sigma_{i/2}^{\omega}f^{*}\|=\|f\|$, so it is an
antilinear isometry on $\mathcal{A}$. Since $\mathcal{A}$ is dense
in $\mathcal{H}_{r}$ this operator can be extended by continuity
to the whole Hilbert space.
\end{proof}
To extend $\tilde{J}$ from $\mathcal{H}_{r}$ to $\mathcal{H}=\mathcal{H}_{r}\otimes\mathbb{C}^{2}$
we introduce the operator $\mathcal{J}=C\tilde{J}$, where $C=i\Gamma^{0}$
is the same matrix as in the commutative case. Now we only have to
check the various properties satisfied by this operator, which are
given in the following.
\begin{prop}
The operator $\mathcal{J}$ is an antilinear isometry on $\mathcal{H}$.
Moreover it satisfies the following properties:
\begin{enumerate}
\item $\mathcal{J}^{2}=1$,
\item $\mathcal{J}\mathcal{D}=-\Delta_{\omega}^{-1}\mathcal{D}\mathcal{J}$,
\item $[\pi(f),\mathcal{J}\pi(g^{*})\mathcal{J}^{-1}]=0$,
\item $\mathcal{J}\chi=-\chi\mathcal{J}$,
\item \textup{$[[\mathcal{D},\pi(f)]_{\sigma},\mathcal{J}\pi(g)\mathcal{J}^{-1}]=0$.}
\end{enumerate}
\end{prop}
\begin{proof}
First we show that $\mathcal{J}$ is an antilinear isometry. When
acting on a spinor $\psi$ we get
\[
\mathcal{J}\left(\begin{array}{c}
\psi_{1}(x)\\
\psi_{2}(x)
\end{array}\right)=\left(\begin{array}{c}
-i\sigma_{i/2}^{\omega}(\psi_{2}^{*})(x)\\
-i\sigma_{i/2}^{\omega}(\psi_{1}^{*})(x)
\end{array}\right)\ .
\]
Then the result follows from the previous lemma and the following
computation
\[
(\mathcal{J}\phi,\mathcal{J}\psi)_{\mathcal{H}}=\sum_{k=1}^{2}\left(-i\sigma_{i/2}^{\omega}(\phi_{k}^{*}),-i\sigma_{i/2}^{\omega}(\psi_{k}^{*})\right)_{\mathcal{H}_{r}}=\sum_{k=1}^{2}(\psi_{k},\phi_{k})_{\mathcal{H}_{r}}=(\psi,\phi)_{\mathcal{H}}\ .
\]
For any $f\in\mathcal{A}$ and $z\in\mathbb{C}$ we have $\sigma_{z}^{\omega}(f)^{*}=\sigma_{\bar{z}}^{\omega}(f^{*})$.
Writing $\tilde{J}\psi=\sigma_{-i/2}^{\omega}(\psi)^{*}$ we have
\[
\tilde{J}^{2}\psi=(\sigma_{-i/2}^{\omega}(\sigma_{-i/2}^{\omega}\psi)^{*})^{*}=\sigma_{i/2}^{\omega}(\sigma_{-i/2}^{\omega}\psi)=\psi\ .
\]
Using this relation the first property is proven as in the commutative
case
\[
\mathcal{J}^{2}=i\Gamma^{0}\tilde{J}i\Gamma^{0}\tilde{J}=(\Gamma^{0})^{2}\tilde{J}^{2}=1\ .
\]
For the second property recall that the Dirac operator is given by
$\mathcal{D}=\Gamma^{\mu}\hat{D}_{\mu}$, where $\hat{D}_{0}=\frac{1}{\lambda}(1-e^{-\lambda\hat{P}_{0}})$
and $\hat{D}_{1}=\hat{P}_{1}$. Using the properties of the $\Gamma$
matrices we obtain $\mathcal{J}\mathcal{D}\mathcal{J}=\Gamma^{\mu}\tilde{J}\hat{D}_{\mu}\tilde{J}$,
as in the commutative case. To compute $\tilde{J}\hat{D}_{\mu}\tilde{J}$
notice that, for any $h\in\mathcal{T}_{\kappa}$, we have
\[
\tilde{J}\rho(h)\tilde{J}\psi=\tilde{J}\rho(h)\sigma_{i/2}^{\omega}(\psi^{*})=(\sigma_{-i/2}^{\omega}\rho(h)\sigma_{i/2}^{\omega}(\psi^{*}))^{*}\ .
\]
But $\sigma_{i/2}^{\omega}$ commutes with $\rho(h)$ for any $h\in\mathcal{T}_{\kappa}$.
Then, using the property of compatibility of the representation with
the star structure $h\triangleright a^{*}=(S(h)^{*}\triangleright a)^{*}$,
we obtain 
\[
\tilde{J}\rho(h)\tilde{J}\psi=(\rho(h)\psi^{*})^{*}=(h\triangleright\psi^{*})^{*}=S(h)^{*}\triangleright\psi=\rho(S(h)^{*})\psi\ .
\]
If we apply this result to $\hat{D}_{\mu}=\rho(D_{\mu})$ we obtain
\[
\begin{split}\tilde{J}\hat{D}_{0}\tilde{J} & =\frac{1}{\lambda}\rho(1-\mathcal{E}^{-1})=-\frac{1}{\lambda}\rho(\mathcal{E}^{-1})\rho(1-\mathcal{E})=-\Delta_{\omega}^{-1}\hat{D}_{0}\ ,\\
\tilde{J}\hat{D}_{1}\tilde{J} & =\rho(-\mathcal{E}^{-1}P_{1})=-\rho(\mathcal{E}^{-1})\rho(P_{1})=-\Delta_{\omega}^{-1}\hat{D}_{1}\ .
\end{split}
\]
Then we obtain $\mathcal{J}\mathcal{D}\mathcal{J}=-\Delta_{\omega}^{-1}\mathcal{D}$,
from which the second property follows. For the third one we notice
that $\mathcal{J}\pi(f^{*})\mathcal{J}=\tilde{J}\pi(f^{*})\tilde{J}$.
We can easily show that $\tilde{J}\pi(f^{*})\tilde{J}$ corresponds
to right multiplication by $\sigma_{i/2}^{\omega}(f)$. Indeed we
have 
\[
\begin{split}\tilde{J}\pi(f^{*})\tilde{J}\psi & =\sigma_{i/2}^{\omega}(f^{*}\star\sigma_{i/2}^{\omega}(\psi^{*}))^{*}=\sigma_{i/2}^{\omega}((\sigma_{i/2}^{\omega}(\psi^{*}))^{*}\star f)\\
 & =\sigma_{i/2}^{\omega}(\sigma_{-i/2}^{\omega}(\psi)\star f)=\psi\star\sigma_{i/2}^{\omega}(f)\ .
\end{split}
\]
Then the property follows from the general fact that right multiplication
commutes with left multiplication. The proof of the fourth property
is identical to the classical case. For the fifth property recall
that $[\mathcal{D},\pi(f)]_{\sigma}=\Gamma^{\mu}\pi(D_{\mu}\triangleright f)$.
Then from the previous property
\[
[[\mathcal{D},\pi(f)]_{\sigma},\mathcal{J}\pi(g)\mathcal{J}^{-1}]=\Gamma^{\mu}[\pi(D_{\mu}\triangleright f),\mathcal{J}\pi(g)\mathcal{J}^{-1}]=0\ .
\]
The proof is complete.
\end{proof}
Several of the requirements for a real structure are trivially satisfied.
This is a consequence of the fact that the Clifford structure is the
same as in the commutative case, that the algebra $\mathcal{A}$ is
involutive and that the representation $\pi$ is by left multiplication.
The first order condition requires the use of the twisted commutator
$[\mathcal{D},\pi(f)]_{\sigma}$, which is indeed the natural choice
for the case of quantum groups \cite{quantum-twisted}. In our case
this condition is automatically satisfied, due to the structure of
our Dirac operator. Finally we discuss the property involving $\mathcal{J}$
and the Dirac operator $\mathcal{D}$. We have shown that for our
spectral triple this is given by $\mathcal{J}\mathcal{D}=-\Delta_{\omega}^{-1}\mathcal{D}\mathcal{J}$.
We have seen that, in checking this property, the role of the antipode
was crucial. Since the antipode of $P_{1}$ is not trivial, we immediately
understand why we can not obtain $\mathcal{J}\mathcal{D}=-\mathcal{D}\mathcal{J}$.